\documentclass{article}

\usepackage{geometry}
\usepackage{color}

\usepackage{amsmath}
\usepackage{amssymb}
\usepackage{amsthm}
\theoremstyle{definition}\newtheorem{dfn}{Definition}[section]
\theoremstyle{plain}
\theoremstyle{plain}\newtheorem{thm}{Theorem}
\theoremstyle{plain}\newtheorem{lem}{Lemma}
\theoremstyle{plain}\newtheorem{claim}{Claim}
\theoremstyle{plain}\newtheorem{cor}{Corollary}
\theoremstyle{plain}
\theoremstyle{plain}

\usepackage{mathtools}
\usepackage{xcolor}
\usepackage{graphicx}
\usepackage{float}
\usepackage[figurewithin=none]{caption}

\usepackage{authblk}

\usepackage{array}

\usepackage{pdfpages}

\usepackage{enumitem}

\newcommand\hn{\textsc{HN}}
\newcommand\ruhn{\textsc{RUHN}}
\newcommand\uhn{\textsc{UHN}}
\newcommand\tc{\textsc{TC}}
\newcommand\utc{\textsc{UTC}}
\newcommand\ndp{\textsc{NDP}}
\newcommand\maf{\textsc{MAF}}

\usepackage{complexity}
\newlang{\rHYB}{HYB}
\newlang{\uHYB}{HYB^\ast}

\usepackage{tikz}
\usetikzlibrary{calc,shapes,snakes}

\newif\ifcomment\commentfalse
\def\commentON{\commenttrue}

\long\outer\def\bc#1\ec{{\ifcomment \sloppy  \textcolor{red}{
{ {#1}} }\fi }}

\long\outer\def\BC#1\EC{{\ifcomment \sloppy \par \textcolor{blue}{\#  \dotfill
{\textsc{#1}} \dotfill \#} \par \fi }}

\commentON

\title{\textbf{On unrooted and root-uncertain variants of several well-known phylogenetic network problems}}
\date{}

\author[1]{Leo van Iersel \thanks{l.j.j.v.iersel@gmail.com}}
\author[2]{Steven Kelk \thanks{steven.kelk@maastrichtuniversity.nl}}
\author[2]{Georgios Stamoulis \thanks{georgios.stamoulis@maastrichtuniversity.nl}}
\author[3,4]{Leen Stougie \thanks{l.stougie@vu.nl}}
\author[2]{Olivier Boes \thanks{olivier.boes@protonmail.com}}
\affil[1]{Delft Institute of Applied Mathematics, Delft University of Technology, Delft, The Netherlands}
\affil[2]{Department of Data Science and Knowledge Engineering (DKE), Maastricht University, Maastricht, The Netherlands}
\affil[3]{CWI and Operations Research, Dept. of Economics and Business Administration, Vrije Universiteit, Amsterdam, The Netherlands}
\affil[4]{INRIA-team ERABLE, France}

\begin{document}
\maketitle

\begin{abstract}
\noindent The hybridization number problem  requires us to embed a set of binary rooted phylogenetic trees into a binary rooted phylogenetic network such that the number of
nodes with indegree two is minimized. However, from a biological point of view accurately inferring the root location in a phylogenetic tree is notoriously difficult and poor
root placement can artificially inflate the hybridization number.  To this end we study a number of relaxed variants of this problem. We start by showing that the fundamental
problem of determining whether an \emph{unrooted} phylogenetic network displays (i.e. embeds) an \emph{unrooted} phylogenetic tree, is NP-hard. On the positive side we show
that this problem is FPT in reticulation number. In the rooted case the corresponding FPT result is trivial, but here we require more subtle argumentation. Next we show that
the hybridization number problem for unrooted networks (when given two unrooted trees) is equivalent to the problem of computing the Tree Bisection and Reconnect (TBR)
distance of the two unrooted trees. In the third part of the paper we consider the ``root uncertain'' variant of hybridization number. Here we are free to choose the root
location in each of a set of unrooted input trees such that the hybridization number of the resulting rooted trees is minimized. On the negative side we show that this
problem is APX-hard. On the positive side, we show that the problem is FPT in the hybridization number, via kernelization, for any number of input trees.
\end{abstract}

\pagebreak

\section{Introduction}

Within the field of phylogenetics the evolutionary history of a set of contemporary species $X$, known as \textit{taxa}, is usually modelled as a tree where the leaves are
bijectively labelled by $X$. One of the central challenges in phylogenetics is to accurately infer this history given only measurements  on $X$ (e.g. one string of DNA per
species in $X$) and to this end many different optimality criteria have been proposed \cite{felsenstein2004inferring,SempleSteel2003}. One issue is that algorithms which
construct evolutionary trees (henceforth: phylogenetic trees) usually produce \emph{unrooted} phylogenetic trees as output i.e. trees in which the direction of evolution is not
specified and thus the notion of ``common ancestor'' is not well-defined. Nevertheless, biologists are primarily interested in \emph{rooted} trees \cite{davidbook}, where the
root, and thus the direction of evolution, is specified. In practice this problem is often addressed by solving the tree-inference and root-inference problem simultaneously,
using a so-called ``outgroup'' \cite{phylogenetics-handbook}. However, this process is prone to error (see \cite{Wilberg01072015} for a recent case-study) and disputes over
rooting location are prominent in the literature (see e.g. \cite{Drew01052014}).

Moreover,  in recent years there has been growing interest in algorithms that construct rooted phylogenetic \emph{networks} \cite{HusonRuppScornavacca10}, essentially the
generalization of rooted phylogenetic trees to rooted directed acyclic graphs. One popular methodology is to construct phylogenetic networks by merging sets of trees according
to some optimality criterion \cite{Nakhleh2009ProbSolv,Huson2011}. For example, in the \textsc{Hybridization Number (\hn)} problem we are given a set of rooted phylogenetic
trees as input and we are required  to topologically embed them into a network $N=(V,E)$
 such that the reticulation number $r(N) = |E| - (|V|-1)$ is minimized; the minimum value thus obtained is known as the hybridization number of the input trees.
 This problem is NP-hard and APX-hard \cite{bordewich} and has similar (in)approximability properties to the classical problem
\textsc{Directed Feedback Vertex Set (DFVS)} \cite{approximationHN}, which is not known to be in APX (i.e. it is not known whether it permits constant-factor polynomial-time
approximation algorithms). On a more positive note, there has been considerable progress on developing fixed parameter tractable (FPT) algorithms for {\hn}. Informally, these
are algorithms which solve {\hn} in time $O( f(k) \cdot \text{poly}(n) )$ where $n$ is the size of the input, $k$ is the hybridization number of the input trees and $f$ is some
computable function that only depends on $k$. FPT algorithms have the potential to run quickly for large $n$, as long as $k$ is small (see \cite{downey2013fundamentals} for an
introduction), and they can be highly effective in applied phylogenetics (see e.g. \cite{Whidden2014,lv2013practical,Gramm2008}). In \cite{sempbordfpt2007} it was proven that
{\hn} is FPT (in the hybridization number) for two input trees  and in recent years the result has been generalized in a number of directions (see \cite{vanIersel20161075} and
the references therein for a recent overview).

One modelling issue with {\hn} is that a poor and/or inconsistent choice of the root location in the input trees can artificially inflate the hybridization number, and this in
turn can (alongside other methodological errors) be misinterpreted as evidence that reticulate evolutionary phenomena such as horizontal gene transfer are abundant
\cite{wendel1998phylogenetic,davidbook}. To take a simple example, consider two identical unrooted trees on a set $X$ of $n$ taxa which \emph{should}, in principle, be rooted
in the same place, so the hybridization number should be 0. If, however, they are rooted in different places due to methodological error, the hybridization number will be at
least 1, and in the worst case can rise to $n-2$. The effect is reinforced as the number of trees in the input increases.

To this end, in this article we study a number of variations of {\hn} (and related decision problems) in which the root has a relaxed role, or no role whatsoever. The first
major part of the article is Section \ref{UTC} in which we analyse the \textsc{Unrooted Tree Containment (UTC)} problem. This is simply the problem of determining whether an
unrooted phylogenetic network $N$ has an unrooted phylogenetic tree $T$ topologically embedded within it. (Following \cite{GBP2012}, an \emph{unrooted phylogenetic network} is
simply a connected, undirected graph where every internal node has degree 3 and the leaves, as usual, are bijectively labelled by $X$).  The rooted version of this problem has
received extensive interest \cite{ISS2010b,bordewich2016reticulation,gunawan2016locating} and, although NP-hard \cite{clustercontainment}, permits a trivial FPT algorithm,
parameterized by the reticulation number of $N$. Here we show that UTC is also NP-hard, addressing a number of technicalities that do not emerge in the rooted case, and FPT in
the reticulation number of $N$. However, here the FPT algorithm is not trivial. We describe a linear kernel based on contracting common chains and subtrees, and a
bounded-search branching algorithm with running time $O(4^k n^2)$, where $k$ is the reticulation number of the network and $n$ is the number of nodes in the network.

In Section~\ref{sec:uhn}, a comparatively short section, we consider the \textsc{Unrooted Hybridization Number (UHN)} problem, where both the input  trees and the output
network are \textit{unrooted}. In this section we restrict our attention to the case when the input has exactly two trees $T_1$ and $T_2$ and we simply ask to find an unrooted
network that displays them both such that the reticulation number of the network is minimized. Consider for example the case of Figure \ref{fig:caterpillars} where we are given
two \textit{unrooted} trees $T_1, T_2$ as input. $N_u$ is a network that displays them both such that $r(N_u) = 1$ and this is optimal. Slightly surprisingly we show that for
UHN the minimum reticulation number of any network that contains both $T_1$ and $T_2$, is equal to the \textsc{Tree Bisection and Reconnection (TBR)} distance of $T_1$ and
$T_2$, which (as is well-known) in turn is equal to the size of an optimum solution to the \textsc{Maximum Agreement Forest (MAF)} problem, minus 1 \cite{AllenSteel2001}.
Hence, the {\uhn} problem on two trees immediately inherits both negative and positive results about TBR/MAF: NP-hardness on one hand, but constant-factor polynomial-time
approximation algorithms and FPT algorithms on the other. This shows that, from an approximation perspective, {\uhn} might be strictly easier than its rooted counterpart which,
as mentioned earlier, might not be in APX at all. It also means that {\uhn} benefits from ongoing, intensive research into {\maf} \cite{whidden2013fixed,
chen2015parameterized,chen2015approximating,bordewich2016fixed}.

In the second major part of the article, Section \ref{sec:ruhn}, we consider the \textsc{Root Uncertain Hybridization Number} (\ruhn) problem. Here the input is a set of
\emph{unrooted} binary trees and we are  to choose the root location of each tree, such that the reticulation number is minimized. See again Figure \ref{fig:caterpillars}. In
contrast with UHN, if we have to root each of $T_1,T_2$ then the minimum reticulation number is 2 and this is achieved by the rooted network $N_r$. This simple example also
shows that UHN can be strictly smaller than RUHN, a point we will elaborate on in the preliminaries. Biologically speaking, RUHN is the most relevant problem we study because
it explicitly acknowledges the fact that the input unrooted trees need to be rooted in some way. This highlights the fact that a root exists, but its location is uncertain and
we would like to infer the root locations such that the reticulation number of a network that displays them all is minimized. On the negative side we show that this problem,
which was explored experimentally in \cite{Whidden2014}, is already NP-hard and APX-hard for two trees. On the positive side,
we show that the problem is FPT (in the hybridization number) for any number of trees, giving a quadratic-sized kernel and discussing how an exponential-time algorithm can be
obtained for solving the kernel. Similar ideas were introduced for the rooted variant in \cite{ierselLinz2013}. Finally, in Section~\ref{sec:conclusions} we conclude with a
number of open questions and future research directions.

\begin{figure}
\centering
\begin{tikzpicture}[x=.7cm, y=.7cm]
\tikzset{lijn/.style={ultra thick}}
\tikzset{stippellijn/.style={ultra thick, dotted}}
\draw[very thick, fill, radius=0.1] (5.5,1.5) circle;
\draw[lijn] (5.5,1.5) -- (6,1);
\draw (5.2,1.5) node {$a$};
\draw[very thick, fill, radius=0.1] (5.5,0.5) circle;
\draw[lijn] (5.5,0.5) -- (6,1);
\draw (5.2,0.5) node {$b$};
\draw[very thick, fill, radius=0.1] (6,1) circle;
\draw[lijn] (6,1) -- (7,1);
\draw[very thick, fill, radius=0.1] (7,1) circle;
\draw[lijn] (7,1) -- (8,1);
\draw[very thick, fill, radius=0.1] (8,1) circle;
\draw[lijn] (8,1) -- (9,1);
\draw[very thick, fill, radius=0.1] (9,1) circle;
\draw[very thick, fill, radius=0.1] (9.5,0.5) circle;
\draw (9.8,0.5) node {$e$};
\draw[lijn] (9,1) -- (9.5,0.5);
\draw[very thick, fill, radius=0.1] (9.5,1.5) circle;
\draw[lijn] (9,1) -- (9.5,1.5);
\draw (9.8,1.5) node {$f$};
\draw[very thick, fill, radius=0.1] (7,.5) circle;
\draw[lijn] (7,1) -- (7,.5);
\draw (7,0) node {$c$};
\draw[very thick, fill, radius=0.1] (8,.5) circle;
\draw[lijn] (8,1) -- (8,.5);
\draw (8,.1) node {$d$};
\draw (7.5,-0.5) node {$T_1$};
\draw[very thick, fill, radius=0.1] (5.5,-1.5) circle;
\draw[lijn] (5.5,-1.5) -- (6,-2);
\draw (5.2,-1.5) node {$c$};
\draw[very thick, fill, radius=0.1] (5.5,-2.5) circle;
\draw[lijn] (5.5,-2.5) -- (6,-2);
\draw (5.2,-2.5) node {$b$};
\draw[very thick, fill, radius=0.1] (6,-2) circle;
\draw[lijn] (6,-2) -- (7,-2);
\draw[very thick, fill, radius=0.1] (7,-2) circle;
\draw[lijn] (7,-2) -- (8,-2);
\draw[very thick, fill, radius=0.1] (8,-2) circle;
\draw[lijn] (8,-2) -- (9,-2);
\draw[very thick, fill, radius=0.1] (9,-2) circle;
\draw[very thick, fill, radius=0.1] (9.5,-2.5) circle;
\draw (9.8,-2.5) node {$e$};
\draw[lijn] (9,-2) -- (9.5,-2.5);
\draw[very thick, fill, radius=0.1] (9.5,-1.5) circle;
\draw[lijn] (9,-2) -- (9.5,-1.5);
\draw (9.8,-1.5) node {$d$};
\draw[very thick, fill, radius=0.1] (7,-2.5) circle;
\draw[lijn] (7,-2) -- (7,-2.5);
\draw (7,-2.9) node {$a$};
\draw[very thick, fill, radius=0.1] (8,-2.5) circle;
\draw[lijn] (8,-2) -- (8,-2.5);
\draw (8,-2.9) node {$f$};
\draw (7.5,-3.5) node {$T_2$};
\end{tikzpicture}
\hspace{.5cm}
\begin{tikzpicture}[x=.7cm, y=.7cm]
\tikzset{lijn/.style={ultra thick}}
\tikzset{stippellijn/.style={ultra thick, dotted}}
\draw[very thick, fill, radius=0.1] (1,0) circle;
\draw[very thick, fill, radius=0.1] (2,0) circle;
\draw[very thick, fill, radius=0.1] (0,1) circle;
\draw[very thick, fill, radius=0.1] (1,2) circle;
\draw[very thick, fill, radius=0.1] (2,2) circle;
\draw[very thick, fill, radius=0.1] (3,1) circle;
\draw[lijn] (1,0) -- (2,0);
\draw[lijn] (2,0) -- (3,1);
\draw[lijn] (3,1) -- (2,2);
\draw[lijn] (2,2) -- (1,2);
\draw[lijn] (1,2) -- (0,1);
\draw[lijn] (0,1) -- (1,0);
\draw[very thick, fill, radius=0.1] (-.5,1) circle;
\draw[lijn] (-.5,1) -- (0,1);
\draw (-.8,1) node {$a$};
\draw[very thick, fill, radius=0.1] (1,2.5) circle;
\draw[lijn] (1,2.5) -- (1,2);
\draw (1,2.9) node {$b$};
\draw[very thick, fill, radius=0.1] (2,2.5) circle;
\draw[lijn] (2,2.5) -- (2,2);
\draw (2,2.9) node {$c$};
\draw[very thick, fill, radius=0.1] (1,-.5) circle;
\draw[lijn] (1,-.5) -- (1,0);
\draw (1,-.9) node {$f$};
\draw[very thick, fill, radius=0.1] (2,-.5) circle;
\draw[lijn] (2,-.5) -- (2,0);
\draw (2,-.9) node {$e$};
\draw[very thick, fill, radius=0.1] (3.5,1) circle;
\draw[lijn] (3.5,1) -- (3,1);
\draw (3.8,1) node {$d$};
\draw (1.5,-2.5) node {$N_u$};
\end{tikzpicture}
\hspace{.5cm}
\begin{tikzpicture}[x=.7cm, y=.7cm]
\tikzset{lijn/.style={ultra thick}}
\tikzset{stippellijn/.style={ultra thick, dotted}}
\tikzset{pijl/.style={ultra thick,->}}
\draw[very thick, fill, radius=0.1] (0,7) circle;
\draw[pijl] (0,7) -- (-1,6);
\draw[pijl] (0,7) -- (1,6);
\draw[very thick, fill, radius=0.1] (-1,6) circle;
\draw[pijl] (-1,6) -- (-2,5);
\draw[very thick, fill, radius=0.1] (-2,5) circle;
\draw[pijl] (-2,5) -- (-3,4);
\draw[very thick, fill, radius=0.1] (-3,4) circle node[below] {$a$};
\draw[pijl] (-2,5) -- (-2,4);
\draw[very thick, fill, radius=0.1] (-2,4) circle;
\draw[pijl] (-2,4) -- (-2,3);
\draw[very thick, fill, radius=0.1] (-2,3) circle node[below] {$b$};
\draw[pijl] (-2,4) -- (-1,3);
\draw[pijl] (-1,6) -- (-1,3);
\draw[very thick, fill, radius=0.1] (-1,3) circle;
\draw[pijl] (-1,3) -- (-1,2);
\draw[very thick, fill, radius=0.1] (-1,2) circle node[below] {$c$};
\draw[very thick, fill, radius=0.1] (1,6) circle;
\draw[pijl] (1,6) -- (1,3);
\draw[very thick, fill, radius=0.1] (1,3) circle;
\draw[pijl] (1,6) -- (2,5);
\draw[very thick, fill, radius=0.1] (2,5) circle;
\draw[pijl] (2,5) -- (3,4);
\draw[very thick, fill, radius=0.1] (3,4) circle node[below] {$d$};
\draw[pijl] (2,5) -- (2,4);
\draw[very thick, fill, radius=0.1] (2,4) circle;
\draw[pijl] (2,4) -- (2,3);
\draw[very thick, fill, radius=0.1] (2,3) circle node[below] {$e$};
\draw[pijl] (2,4) -- (1,3);
\draw[pijl] (1,3) -- (1,2);
\draw[very thick, fill, radius=0.1] (1,2) circle node[below] {$f$};
\draw (0,1) node {$N_r$};
\end{tikzpicture}
\caption{\label{fig:caterpillars} \emph{Two unrooted  trees~$T_1,T_2$, an unrooted network~$N_u$ with reticulation number~1 that displays~$T_1$ and~$T_2$ and a rooted
network~$N_r$ that displays rootings of~$T_1$ and~$T_2$ and has reticulation number~2. $N_u$ is an optimal solution to the UHN problem, while $N_r$ is an optimal solution to the
RUHN problem.}}
\end{figure}
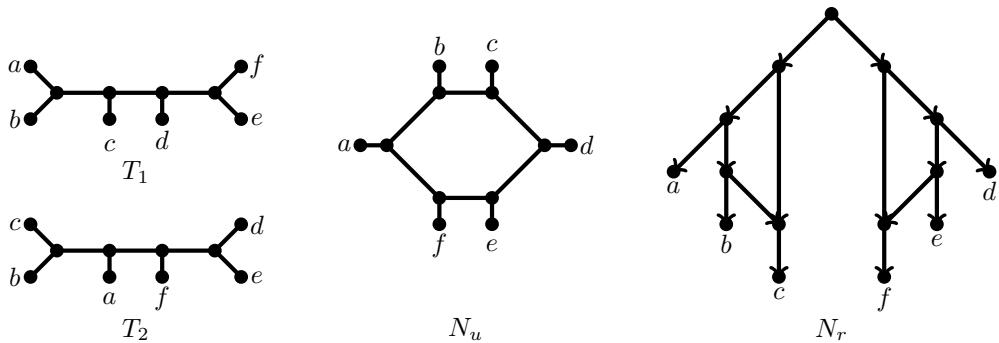

We begin with a section dedicated to preliminaries in which we formally define all the models studied in this paper and briefly discuss their differences.

\section{Preliminaries}

A \emph{rooted binary phylogenetic network} $N=(V,E)$ on a set of leaf-labels (also known as \emph{taxa}) $X$ (where $|X| \geq 2$) is a directed acyclic graph (DAG)  in which
the leaves are bijectively labelled by $X$, there is exactly one \emph{root} (a node with indegree 0 and outdegree 2), and all nodes are either \emph{tree} nodes (indegree 1,
outdegree 2) or \emph{reticulation nodes} (indegree 2, outdegree 1). The \emph{reticulation number} $r(N)$ of $N$ is defined as $|E|-(|V|-1)$, which is equal to the number of
reticulation nodes in $N$. A rooted binary phylogenetic network $N$ which has $r(N)=0$ is simply called a \emph{rooted binary phylogenetic tree}. Two rooted binary phylogenetic
networks $N_1$ and $N_2$ on $X$ are said to be isomorphic if there exists an isomorphism between $N_1$ and $N_2$ which is the identity on $X$.

Similarly, an \emph{unrooted binary phylogenetic network} on $X$ is simply a connected, undirected graph $N=(V,E)$ with $|X|$ nodes of degree 1 (i.e., leaves), labelled
bijectively by $X$, and all other nodes are of degree 3. Although notions of indegree and outdegree do not apply here, reticulation number can still naturally be defined as
$r(N) = |E|-(|V|-1)$. An unrooted binary phylogenetic tree is an unrooted binary phylogenetic network with $r(N)=0$. See Figure \ref{fig:caterpillars} for examples of rooted
and unrooted networks.

Throughout the article we will occasionally refer to \emph{caterpillars}. For  $n \geq 4$ an \emph{unrooted caterpillar} $(x_1, ..., x_n)$ is the unrooted binary phylogenetic
tree constructed as follows: it consists of a central path $(p_2, ..., p_{n-1})$ with a single taxon $x_i$ adjacent to each $p_i$ (for $3 \leq i \leq n-2$), two taxa $x_1$ and
$x_2$ adjacent to $p_2$ and two taxa $x_{n-1}$ and $x_n$ adjacent to $p_{n-1}$. The two trees shown in Figure \ref{fig:caterpillars} are both unrooted caterpillars with $n=6$.
A \emph{rooted caterpillar} is obtained by subdividing the edge $\{p_2,x_1\}$, taking the newly created node $p_1$ as the root and directing all edges away from it.

We say that a rooted binary phylogenetic network  $N$ on $X$ \emph{displays} a rooted binary phylogenetic tree $T$ on $X$ if $T$ can be obtained from a subtree $T'$ of $N$ by
suppressing nodes with indegree 1 and outdegree 1. Similarly, an unrooted binary phylogenetic network $N$ on $X$ \emph{displays} an unrooted  binary phylogenetic tree $T$ on
$X$ if $T$ can be obtained from a subtree $T'$ of $N$ by suppressing nodes of degree 2. In both cases we say that $T'$ is an \emph{image} of $T$.

Consider the following problem, which has been studied extensively, and its unrooted variant.\\
\\
\textbf{Problem: } \textsc{Hybridization Number (HN)}\\
\textbf{Input: } A set $\mathcal{T}$ of rooted binary phylogenetic trees on the same set of
taxa $X$.\\
\textbf{Output: } A rooted binary phylogenetic network $N$ on $X$ such that, for each
$T \in \mathcal{T}$, $N$ displays $T$.\\
\textbf{Goal: } Minimize $r(N)$.\\
\\
The minimum value of $r(N)$ thus obtained we denote by $h^{r}(\mathcal{T})$ and this is also called the \emph{hybridization number} of $\mathcal{T}$. Note that feasible
solutions to {\hn} should include certificates verifying that the input trees are actually displayed by $N$. A certificate in this case is usually an image of each tree
embedded in $N$. Determining whether a rooted network displays a rooted tree - this is the \textsc{Tree Containment} problem - is NP-hard \cite{clustercontainment}, hence the
need for these certificates. (Certificates are thus also required for {\ruhn}, to be defined in due course).

\textsc{HN} is APX-hard (and thus NP-hard) \cite{bordewich} but FPT in parameter $h^{r}(\mathcal{T})$ \cite{ierselLinz2013}. That is, answering the question ``Is
$h^{r}(\mathcal{T}) \leq k$?'' can be answered in time $O( f(k) \cdot \text{poly}(n) )$ where $f$ is a computable function that only depends on $k$ and $n$ is the size of the
input to \textsc{HN}.
It is well-known that $h^{r}(\mathcal{T}) = 0$ if and only if all the trees in $\mathcal{T}$ are isomorphic, which can easily be checked in polynomial time \cite{SempleSteel2003}.\\
\\
\textbf{Problem: } \textsc{Unrooted Hybridization Number (UHN)}\\
\textbf{Input: } A set $\mathcal{T}$ of unrooted binary phylogenetic trees on the same set of
taxa $X$. \\
\textbf{Output: } An unrooted binary phylogenetic network $N$ on $X$ such that, for each
$T \in \mathcal{T}$, $N$ displays $T$. \\
\textbf{Goal: } Minimize $r(N)$. \\
\\
We write $h^u(\mathcal{T})$ to denote the minimum value of $r(N)$ thus obtained.  It is natural to ask: do feasible solutions to {\uhn} require, as in the rooted case,
certificates verifying that the input trees are displayed by the network? This motivates our study of the following problem, which we will start with in Section~\ref{UTC}. \\
\\
\textbf{Problem: } \textsc{Unrooted Tree Containment} (\utc)\\
\textbf{Input: } An unrooted binary phylogenetic network $N$ and an unrooted binary phylogenetic tree $T$, both  on $X$.\\
\textbf{Question: } Does $N$ display $T?$\\
\\
Finally, we will consider the variant in which we require the unrooted input trees to be rooted. A rooted binary phylogenetic tree $T'$ on $X$ is a \emph{rooting} of an
unrooted binary phylogenetic tree $T$ if $T'$ can be obtained by subdividing an edge of $T$ with a new node $u$ and directing all edges away from $u$. We say that $T$ is the
\emph{unrooting} of $T'$, denoted $U(T')$.  \\
\\
\textbf{Problem: } \textsc{Root Uncertain Hybridization Number (RUHN)}\\
\textbf{Input: } A set $\mathcal{T}$ of unrooted binary phylogenetic trees on the same set of
taxa $X$.\\
\textbf{Output: } A root location (i.e. an edge) of each tree in $T \in \mathcal{T}$ (which induces a set of rooted binary phylogenetic trees $\mathcal{T'}$ on $X$) and a
rooted binary phylogenetic network  $N$ on $X$ such that, for each $T' \in \mathcal{T'}$, $N$ displays
$T'$.\\
\textbf{Goal: } Minimize $r(N)$.\\
\\
The minimum value of $r(N)$ obtained is denoted $h^{ru}(\mathcal{T})$ and is called the \emph{root-uncertain hybridization number} of $\mathcal{T}$. Note that if $\mathcal{T}$
is a set of rooted binary phylogenetic trees and $\mathcal{T^{*}}$ is the set of unrooted counterparts of $\mathcal{T}$ - that is, $\mathcal{T^{*}} = \{ U(T) | T \in
\mathcal{T}\}$  - then $h^{ru}(\mathcal{T^{*}})$ can differ significantly from $h^{r}(\mathcal{T})$. For example, if $\mathcal{T}$ consists of two rooted caterpillars on the
same set of $n$ taxa, but with opposite orientation, then $h^{r}(\mathcal{T}) = n-2$ whilst $h^{ru}(\mathcal{T^{*}}) = 0.$ More generally, a little thought should make it clear
that on a set $\mathcal{T}$ of binary rooted trees and the set $\mathcal{T^*}$ of their corresponding unrooted versions, we have:
\begin{eqnarray}
\label{eq:sequence}
  h^u(\mathcal{T^*}) \leq h^{ru}(\mathcal{T^*}) \leq h^{r}(\mathcal{T}).
\end{eqnarray}

It is possible to say more about this inequality chain. Let $\mathcal{T^*}$ be the two unrooted binary trees $T_1$ and $T_2$ shown in Figure \ref{fig:caterpillars}. It is easy
to see that $h^{u}(\mathcal{T^*})=1$: we simply arrange the taxa in a circle with circular ordering $a,b,c,d,e,f$ (see $N_u$ in Figure \ref{fig:caterpillars}). However, as can
be verified by case analysis (or using the ``re-root by hybridization number''
 functionality in \textsc{Dendroscope} \cite{huson2012dendroscope}), $h^{ru}(\mathcal{T^*})=2$. Moreover, let $\mathcal{T}$ be the two rooted trees
obtained by rooting the first tree on the edge entering $a$, and the second tree on the edge entering $e$. It can be verified that $h^{r}(\mathcal{T})=3$. Hence,  $\mathcal{T}$
is an example when both inequalities in (\ref{eq:sequence}) are simultaneously strict.


\section{The Tree Containment problem  on unrooted networks and trees}\label{UTC}

\noindent Given a rooted binary phylogenetic network $N = (V,E)$ on $X$ and a rooted binary phylogenetic tree $T$ also on $X$ it is trivial to determine in time $O( 2^k \cdot \text{poly}(n) )$ whether
$N$ displays $T$, where $k = r(N) = |E|- (|V|-1)$ and $n=|V|$. This is because, for each of the $k$ reticulation nodes, we can simply guess which of its two incoming edges are on the image of $T$.
Here we consider the natural unrooted analogue of the problem where both $N$ and $T$ are unrooted.

We show that the question whether
$N$ displays $T$ is NP-hard, but FPT when parameterized by $k=r(N)=|E|-(|V|-1)$. Note that, unlike for the rooted case, an FPT result here is not trivial,  since the notion ``reticulation node'' no
longer has any meaning.

\subsection{The hardness of \textsc{Unrooted Tree Containment (UTC)}}
\begin{thm}
UTC is NP-hard.
\end{thm}
\begin{proof}
We reduce from the problem \textsc{Node Disjoint Paths On Undirected Graphs} (\ndp). The reduction is similar in spirit to the reduction given in \cite{clustercontainment},
where the hardness of tree containment on \emph{rooted} networks was proven by reducing from {\ndp} on \emph{directed} graphs. However, our reduction has to deal with a number
of subtleties specific to the case of unrooted trees and networks.

{\ndp} is defined as follows. We are given an undirected graph $G=(V,E)$ and a multiset of unordered pairs of nodes $W=\{ \{s_1, t_1\}, \ldots, \{s_k, t_k\} \}$, where for each
$i$, $s_i \neq t_i$. Note that we do not assume a distinction between the $s$ nodes and the $t$ nodes (we refer to them together as \emph{terminals}), and the same pair can
appear several times. The question is: do there exist $k$ paths $P_i$ ($1 \leq i \leq k$) such that $P_i$ connects $s_i$ to $t_i$, and such that the $P_i$ are mutually
node-disjoint?

The literature is somewhat ambiguous about whether endpoints of the paths are allowed to intersect, and of course this is a necessary condition if we are to
allow some terminal to appear in more than one pair in $W$. We posit as few restrictions as possible on the input - specifically, we allow each terminal to be in multiple pairs
- and then show that this can be reduced to a more restricted instance. We do however make use of the fact that {\ndp} remains hard on cubic graphs\footnote{ This follows from
\cite{richards1984complexity}. In that article the hardness of {\ndp} is proven for undirected graphs of maximum degree 3, but using standard gadgets nodes of degree 1 or 2 can
easily be upgraded to degree 3. See also \cite{schrijver2002combinatorial}, p. 1225 for a related discussion.}, and assume henceforth that $G$ is cubic.

We start by first reducing the cubic {\ndp} instance
$(G,W)$ to a new instance $(G',W')$ where $G'$ has maximum degree 3 and no nodes of degree 2, each terminal appearing within $W'$ is in exactly one pair, and a node of $G'$ is
a terminal if and only if it has degree 1. As usual, the idea is that $(G,W)$ is a YES instance for {\ndp} if and only $(G',W')$ is. The transformation to $(G',W')$ is
straightforward. Observe firstly that in the $(G,W)$ instance each terminal can appear in at most 3 pairs (otherwise it is trivially a NO instance). Depending on whether a
terminal is in 1, 2, or 3 pairs we use a different transformation.

\begin{enumerate}
\item \emph{A terminal is in 3 pairs in $W$: $\{s_i, t_i\}, \{s_j, t_j\}, \{s_k, t_k\}$ where $s_i = s_j= s_k$.} We split the terminal into 3 distinct nodes; see Figure
    \ref{fig:disjoint}(left).
\item \emph{A terminal is in 2 pairs in $W$: $\{s_i, t_i\}, \{s_j, t_j\}$ where $s_i = s_j$.} We split the terminal into 2 distinct nodes; see Figure
    \ref{fig:disjoint}(middle).
\item \emph{A terminal is in 1 pair in $W$: $\{s_i, t_i\}$.} Here we do not split the terminal but we do introduce a new terminal pair $\{p,q\}$; see Figure
    \ref{fig:disjoint}(right). The introduction of $\{p,q\}$ concerns the fact that, prior to the transformation, at most one of the node disjoint paths can intersect with
    node $s_i$. The presence of $\{p,q\}$ ensures that, after transformation, at most one path can intersect with the image of this node. (A simpler transformation is not
    obviously possible, due to the degree restrictions on $G'$).
\end{enumerate}

\begin{figure}
\centering
\begin{tikzpicture}[x=.7cm, y=.7cm]
\tikzset{lijn/.style={ultra thick}}
\tikzset{pijl/.style={ultra thick,->}}

\draw[lijn] (0,1) -- (0,0.1);
\draw[very thick, fill =white, radius=0.1] (0,0) circle node[above right] {$s_i=s_j=s_k$};
\draw[lijn] (-1,0) -- (-0.1,0);
\draw[lijn] (1,0) -- (0.1,0);

\draw[pijl] (0,-1) -- (0,-2);
\draw[pijl] (-2,-1) -- (-3,-2);
\draw[pijl] (2,-1) -- (3,-2);

\draw[lijn] (-6,-3) -- (-6,-4);
\draw[very thick, fill, radius=0.1] (-6,-4) circle;
\draw[lijn] (-6,-4) -- (-6,-5);
\draw[lijn] (-6,-5) -- (-7,-5);
\draw[lijn] (-6,-5) -- (-5,-5);
\draw[very thick, fill, radius=0.1] (-7,-5) circle;
\draw[lijn] (-7,-5) -- (-8,-5);
\draw[very thick, fill, radius=0.1] (-5,-5) circle;
\draw[lijn] (-5,-5) -- (-4,-5);
\draw[lijn] (-7,-5) -- (-7,-6);
\draw[very thick, fill, radius=0.1] (-7,-6) circle;
\draw[lijn] (-5,-5) -- (-5,-6);
\draw[very thick, fill, radius=0.1] (-5,-6) circle;
\draw[lijn] (-7,-6) -- (-7,-7);
\draw[very thick, fill, radius=0.1] (-7,-7) circle;
\draw[lijn] (-5,-6) -- (-5,-7);
\draw[very thick, fill, radius=0.1] (-5,-7) circle;
\draw[lijn] (-5,-6) -- (-7,-7);
\draw[lijn] (-7,-6) -- (-5,-7);
\draw[lijn] (-5,-7) -- (-5,-8);
\draw[very thick, fill, radius=0.1] (-5,-8) circle;
\draw[lijn] (-6,-4) to[out=-60,in=60] (-6,-8);
\draw[very thick, fill, radius=0.1] (-6,-8) circle;
\draw[lijn] (-5,-8) -- (-5,-9);
\draw[very thick, fill, radius=0.1] (-5,-9) circle;
\draw[lijn] (-6,-8) -- (-6,-9);
\draw[very thick, fill, radius=0.1] (-6,-9) circle;
\draw[lijn] (-5,-8) -- (-6,-9);
\draw[lijn] (-6,-8) -- (-5,-9);
\draw[lijn] (-6,-9) -- (-6,-10);
\draw[very thick, fill, radius=0.1] (-6,-10) circle;
\draw[lijn] (-7,-7) -- (-7,-10);
\draw[very thick, fill, radius=0.1] (-7,-10) circle;
\draw[lijn] (-6,-10) -- (-6,-11);
\draw[very thick, fill, radius=0.1] (-6,-11) circle;
\draw[lijn] (-7,-10) -- (-7,-11);
\draw[very thick, fill, radius=0.1] (-7,-11) circle;
\draw[lijn] (-7,-10) -- (-6,-11);
\draw[lijn] (-6,-10) -- (-7,-11);
\draw[lijn] (-7,-11) -- (-7,-12);
\draw[very thick, fill, radius=0.1] (-7,-12) circle node[below] {$s_i$};
\draw[lijn] (-6,-11) -- (-6,-12);
\draw[very thick, fill, radius=0.1] (-6,-12) circle node[below] {$s_j$};
\draw[lijn] (-5,-9) -- (-5,-12);
\draw[very thick, fill, radius=0.1] (-5,-12) circle node[below] {$s_k$};

\draw[very thick, fill=white, radius=0.1] (-6,-5) circle;

\draw[lijn] (0,-3) -- (0,-5);
\draw[lijn] (0,-5) -- (-1,-5);
\draw[lijn] (0,-5) -- (1,-5);
\draw[very thick, fill, radius=0.1] (-1,-5) circle;
\draw[lijn] (-1,-5) -- (-2,-5);
\draw[very thick, fill, radius=0.1] (1,-5) circle;
\draw[lijn] (1,-5) -- (2,-5);
\draw[lijn] (-1,-5) -- (-1,-6);
\draw[very thick, fill, radius=0.1] (-1,-6) circle;
\draw[lijn] (1,-5) -- (1,-6);
\draw[very thick, fill, radius=0.1] (1,-6) circle;
\draw[lijn] (-1,-6) -- (-1,-7);
\draw[very thick, fill, radius=0.1] (-1,-7) circle;
\draw[lijn] (1,-6) -- (1,-7);
\draw[very thick, fill, radius=0.1] (1,-7) circle;
\draw[lijn] (1,-6) -- (-1,-7);
\draw[lijn] (-1,-6) -- (1,-7);
\draw[lijn] (-1,-7) -- (-1,-12);
\draw[very thick, fill, radius=0.1] (-1,-12) circle node[below] {$s_i$};
\draw[lijn] (1,-7) -- (1,-12);
\draw[very thick, fill, radius=0.1] (1,-12) circle node[below] {$s_j$};

\draw[very thick, fill=white, radius=0.1] (0,-5) circle;

\draw[lijn] (6,-3) -- (6,-5);
\draw[very thick, fill, radius=0.1] (6,-4) circle;
\draw[lijn] (6,-4) -- (8,-4);
\draw[very thick, fill, radius=0.1] (8,-4) circle;
\draw[lijn] (8,-4) -- (8,-5);
\draw[very thick, fill, radius=0.1] (8,-5) circle;
\draw[lijn] (8,-5) -- (9,-5);
\draw[lijn] (8,-4) -- (9,-3);
\draw[very thick, fill, radius=0.1] (9,-3) circle node[right] {$q$};
\draw[lijn] (6,-5) -- (5,-5);
\draw[lijn] (6,-5) -- (7,-5);
\draw[very thick, fill, radius=0.1] (5,-5) circle;
\draw[lijn] (5,-5) -- (4,-5);
\draw[very thick, fill, radius=0.1] (7,-5) circle;
\draw[lijn] (7,-5) -- (8,-5);
\draw[lijn] (5,-5) -- (5,-6);
\draw[very thick, fill, radius=0.1] (5,-6) circle;
\draw[lijn] (7,-5) -- (7,-6);
\draw[very thick, fill, radius=0.1] (7,-6) circle;
\draw[lijn] (5,-6) -- (5,-7);
\draw[very thick, fill, radius=0.1] (5,-7) circle;
\draw[lijn] (7,-6) -- (7,-7);
\draw[very thick, fill, radius=0.1] (7,-7) circle;
\draw[lijn] (7,-6) -- (5,-7);
\draw[lijn] (5,-6) -- (7,-7);
\draw[lijn] (5,-7) -- (5,-12);
\draw[very thick, fill, radius=0.1] (5,-12) circle node[below] {$s_i$};
\draw[lijn] (7,-7) -- (7,-12);
\draw[very thick, fill, radius=0.1] (7,-12) circle node[below] {$p$};

\draw[very thick, fill=white, radius=0.1] (6,-5) circle;
\end{tikzpicture}
\caption{\label{fig:disjoint} \emph{Gadgets for obtaining a transformed {\ndp} instance $(G',W')$ where $G'$ has maximum degree 3, no nodes of degree 2, a node has degree 1 if and only if it is a terminal, and each terminal
appears in exactly one pair.}}
\end{figure}
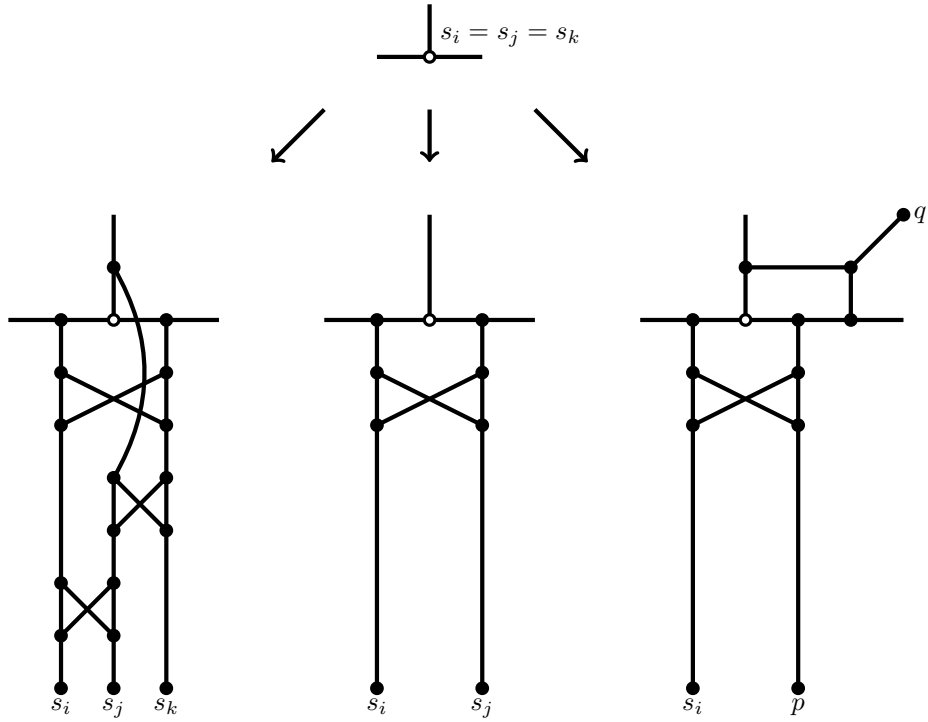

The transformations are applied as often as necessary to obtain the instance $(G',W')$. Let
$W'=\{\{s_1, t_1\}, \ldots, \{s_{k'}, t_{k'}\}\}$. Due to the fact that each terminal now appears in exactly
one pair, we can schematically view the $(G',W')$ instance as shown in Figure \ref{fig:schematic}.

\begin{figure}
\centering
\begin{tikzpicture}[x=.7cm, y=.7cm]
\tikzset{lijn/.style={ultra thick}}
\tikzset{stippellijn/.style={ultra thick, dotted}}

\draw[lijn] (-3.63,-6.63) -- (-5,-4.5);
\draw[very thick, fill, radius=0.1] (-5,-4.5) circle node[above] {$s_{k'}$};

\draw[lijn] (-1,-6) -- (-1,-4.5);
\draw[very thick, fill, radius=0.1] (-1,-4.5) circle node[above] {$s_{k'-1}$};

\node[draw=none,fill=none] at (2,-4.5) {$\ldots$};

\draw[lijn] (3.63,-6.63) -- (5,-4.5);
\draw[very thick, fill, radius=0.1] (5,-4.5) circle node[above] {$s_1$};

\draw (0,-8) node[ellipse, thick, minimum height=3cm,minimum width=8cm,draw, fill=lightgray] {$G'$};

\draw[lijn] (-3.63,-9.63) -- (-5,-11);
\draw[very thick, fill, radius=0.1] (-5,-11) circle node[below] {$t_{k'}$};

\draw[lijn] (-1,-10.1) -- (-1,-11);
\draw[very thick, fill, radius=0.1] (-1,-11) circle node[below] {$t_{k'-1}$};

\node[draw=none,fill=none] at (2,-11) {$\ldots$};

\draw[lijn] (3.63,-9.63) -- (5,-11);
\draw[very thick, fill, radius=0.1] (5,-11) circle node[below] {$t_1$};
\end{tikzpicture}
\caption{\label{fig:schematic} \emph{Schematic representation of the transformed {\ndp} instance $(G',W')$.}}
\end{figure}
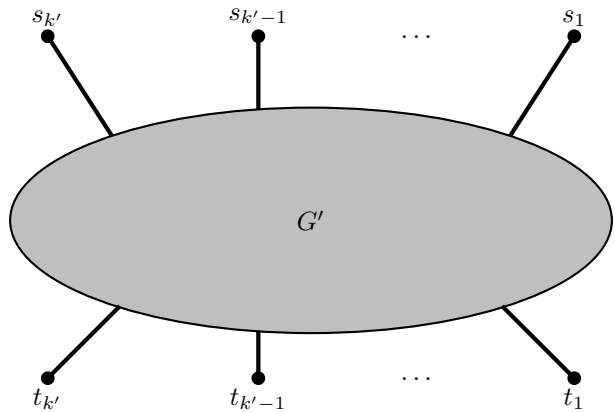

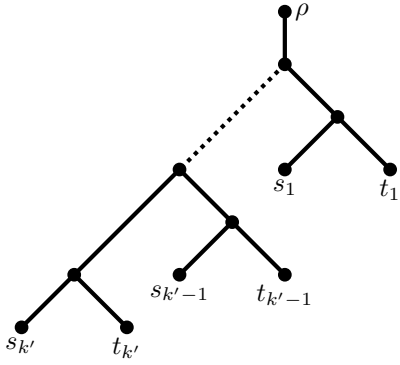
\begin{figure}
\centering
\begin{tikzpicture}[x=.7cm, y=.7cm]
\tikzset{lijn/.style={ultra thick}}
\tikzset{stippellijn/.style={ultra thick, dotted}}
\draw[very thick, fill, radius=0.1] (0,0) circle node[right] {$\rho$};
\draw[very thick, fill, radius=0.1] (0,-1) circle; 
\draw[very thick, fill, radius=0.1] (1,-2) circle; 
\draw[very thick, fill, radius=0.1] (2,-3) circle node[below] {$t_1$};
\draw[very thick, fill, radius=0.1] (0,-3) circle node[below] {$s_1$};
\draw[very thick, fill, radius=0.1] (-2,-3) circle; 
\draw[very thick, fill, radius=0.1] (-1,-4) circle; 
\draw[very thick, fill, radius=0.1] (-2,-5) circle node[below] {$s_{k'-1}$};
\draw[very thick, fill, radius=0.1] (0,-5) circle node[below] {$t_{k'-1}$};
\draw[very thick, fill, radius=0.1] (-4,-5) circle; 
\draw[very thick, fill, radius=0.1] (-5,-6) circle node[below] {$s_{k'}$};
\draw[very thick, fill, radius=0.1] (-3,-6) circle node[below] {$t_{k'}$};
\draw[lijn] (0,0) -- (0,-1);
\draw[lijn] (0,-1) -- (1,-2);
\draw[lijn] (1,-2) -- (0,-3);
\draw[lijn] (1,-2) -- (2,-3);
\draw[stippellijn] (0,-1) -- (-2,-3);
\draw[lijn] (-2,-3) -- (-1,-4);
\draw[lijn] (-1,-4) -- (-2,-5);
\draw[lijn] (-1,-4) -- (0,-5);
\draw[lijn] (-2,-3) -- (-4,-5);
\draw[lijn] (-4,-5) -- (-3,-6);
\draw[lijn] (-4,-5) -- (-5,-6);
\end{tikzpicture}
\caption{\label{fig:treehard} \emph{The tree $T$ used in the reduction of {\ndp} to {\utc}.}}
\end{figure}

\begin{figure}
\centering
\begin{tikzpicture}[x=.7cm, y=.7cm]
\tikzset{lijn/.style={ultra thick}}
\tikzset{stippellijn/.style={ultra thick, dotted}}
\draw[very thick, fill, radius=0.1] (0,0) circle node[right] {$\rho$};
\draw[very thick, fill, radius=0.1] (0,-1) circle; 
\draw[very thick, fill, radius=0.1] (1,-2) circle; 
\draw[very thick, fill, radius=0.1] (0,-3) circle node[below] {$s_1$};
\draw[very thick, fill, radius=0.1] (-2,-3) circle; 
\draw[very thick, fill, radius=0.1] (-1,-4) circle; 
\draw[very thick, fill, radius=0.1] (-2,-5) circle node[below] {$s_{k'-1}$};
\draw[very thick, fill, radius=0.1] (-4,-5) circle; 
\draw[very thick, fill, radius=0.1] (-5,-6) circle node[below] {$s_{k'}$};

\draw (0,-8) node[ellipse, thick, minimum height=3cm,minimum width=8cm,draw, fill=lightgray] {$G'$};

\draw[lijn] (-3.63,-9.63) -- (-5,-11);
\draw[very thick, fill, radius=0.1] (-5,-11) circle node[below] {$t_{k'}$};

\draw[lijn] (-1,-10.1) -- (-1,-11);
\draw[very thick, fill, radius=0.1] (-1,-11) circle node[below] {$t_{k'-1}$};

\node[draw=none,fill=none] at (2,-11) {$\ldots$};

\draw[lijn] (3.63,-9.63) -- (5,-11);
\draw[very thick, fill, radius=0.1] (5,-11) circle node[below] {$t_{1}$};

\draw[lijn] (0,0) -- (0,-1);
\draw[lijn] (0,-1) -- (1,-2);
\draw[lijn] (1,-2) -- (0,-3);
\draw[lijn] (1,-2) -- (3.22,-6.25); 

\draw[stippellijn] (0,-1) -- (-2,-3);
\draw[lijn] (-2,-3) -- (-1,-4);
\draw[lijn] (-1,-4) -- (-2,-5);
\draw[lijn] (-1,-4) -- (0,-5.86); 
\draw[lijn] (-2,-3) -- (-4,-5);
\draw[lijn] (-4,-5) -- (-3.2,-6.25); 
\draw[lijn] (-4,-5) -- (-5,-6);
\end{tikzpicture}
\caption{\label{fig:networkhard} \emph{The network $N$ used in the reduction of {\ndp} to {\utc}.}}
\end{figure}
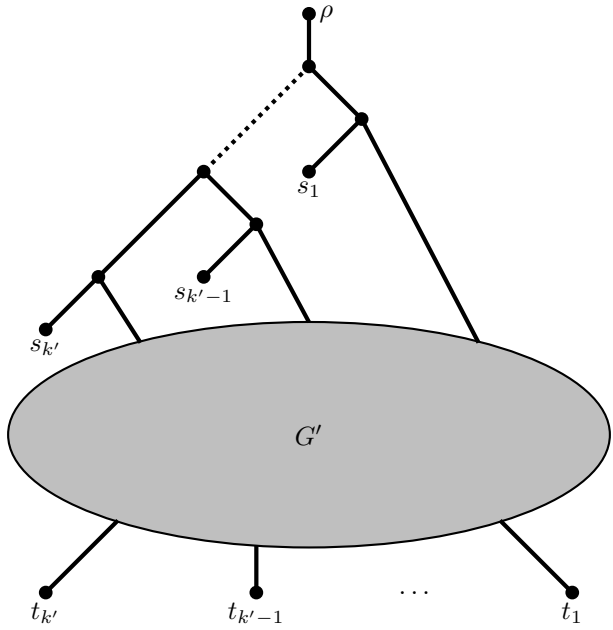

Now, we reduce $(G',W')$ to \textsc{UTC}. Let $T$ be the  unrooted binary phylogenetic tree on $2k'+1$ taxa $X = \{\rho, s_1, t_1, \ldots, s_{k'}, t_{k'} \}$ shown
in Figure \ref{fig:treehard}. The unrooted binary phylogenetic network $N$, also on $X$, is constructed from $(G',W')$ as shown in Figure \ref{fig:networkhard}. It can easily
be verified that $N$ displays $T$ if and only if $(G',W')$ is a YES instance to {\ndp}.
\end{proof}

\pagebreak
\subsection{Unrooted Tree Containment (UTC) parameterized by reticulation number}
\label{subsec:utcfpt}

Recall that the input to {\utc} is an unrooted binary phylogenetic network $N = (V,E)$ and an unrooted binary phylogenetic tree $T$, both  on $X$. In this section we use $n :=
|V|$ to denote the size of the input to {\utc}, which is justified by noticing that $|X| \leq |V|$  and $|V| - 1 \leq |E| \leq (3/2)|V|$ and that $|V|$ can be arbitrarily
larger than $|X|$.

We prove that {\utc} is fixed parameter tractable (FPT) in parameter $r(N)$.  First, we give a linear kernel: we show how to transform in $\text{poly}(n)$ time the instance
$(N,T)$ of ${\utc}$ into a new instance $(N'', T'')$ on $X''$ such that $r(N'') \leq r(N)$, the size of the instance $(N'',T'')$ is at most a linear function of $r(N'')$, and
$N''$ displays $T''$ if and only if $N$ displays $T$. Second, we describe a simple bounded-search branching algorithm to answer {\utc}, and combining these two results
establishes the FPT result. (Note that the second result alone is actually sufficient to establish the FPT result, and could be used without first applying the kernelization
procedure, but the kernelization is of independent interest and can contribute to further speed-up in practice).

We start with some necessary definitions,  which we give in a form somewhat more general than required in this section, so that we can use them in later sections.  Let
$\mathcal{N}$ be a collection of binary, unrooted, phylogenetic networks (all on $X$) and $N_i \in \mathcal{N}$. A  subtree $T$ is called a \textit{pendant} subtree of $N_i$ if
there exists an edge $e$ the deletion of which detaches $T$ from $N_i$. A subtree $T$, which induces a subset of taxa $X' \subset X$, is called \textit{common pendant subtree
of} $\mathcal{N}$ if the following two conditions hold:

\begin{enumerate}
\item $T$ is a pendant subtree  on each $N_i \in \mathcal{N}$ and $N_i | X' = N_j|X'$ for each pair of two distinct networks $N_i, N_j \in \mathcal{N}$. Here by $N_i|X'$ we
    mean the tree which is obtained from $N_i$ by taking the minimum spanning tree on $X'$ and then suppressing any resulting node of degree 2.

\item Let $e_i$ be the edge of network $N_i \in \mathcal{N}$ the deletion of which detaches $T$ from $N_i$ and let $v_i \in e_i$ be the endpoint of $e_i$ ``closest" to the
    taxon set $X'$. Let's say that we root each $N_i|X'$ at $v_i$, thus inducing a rooted binary phylogenetic tree $(N_i|X')^{\rho}$ on $X'$. We require that, for each
    distinct pair of networks $N_i, N_j \in \mathcal{N}$, $(N_i|X')^{\rho} = (N_j|X')^{\rho}$.
\end{enumerate}
The second condition formalizes the idea that the point of contact (root location) of the tree should explicitly be taken into account when determining whether a pendant
subtree is common. (This is consistent with the definition of common pendant subtree elsewhere in the literature).

The above definition is the basis of the following polynomial-time reduction rule which we will use extensively.

\begin{description}
\item[Common Pendant Subtree (CPS) reduction:] Find a maximal common pendant subtree in $\mathcal{N}$. Let $T$ be such a common subtree with at least two taxa and let $X_T$
    be its set of taxa. Clip $T$ from each $N_i \in \mathcal{T}$. Attach a single label $x \notin X$ in place of $T$ on
    each $N_i$. Set $X := (X \setminus X_T) \cup \{x\}$.
\end{description}

\noindent Note that, if all the networks in $\mathcal{N}$ are copies of the same, identical unrooted binary tree on $X$, we adopt the convention that iterated application of
the CPS reduction reduces all the trees to a single taxon $\{x\}$.

Next, let $N$ be an unrooted binary network on $X$. For each taxon $x_i \in X$, let $p_i$ be the unique parent of $x_i$ in $N$. Let $C = (x_1, x_2, \dots, x_t)$ be an ordered sequence
of taxa and let $P = (p_1, p_2, \dots, p_t)$ be the ordered sequence corresponding to their parents. We allow $p_1 = p_2$ or $p_{t-1} = p_t$. If
$P$ is a \textit{path} in $N$  then $C$ is called a \textit{chain} of length $t$.  A chain $C$ is a \textit{common chain} of $\mathcal{N}$ if $C$ is a chain in each $N_i \in
\mathcal{N}$. This brings us to our second polynomial-time reduction rule.

\begin{description}
\item[Common $d$-Chain ($d$-CC) reduction:] Find a maximal common chain $C = (x_1, \ldots, x_t)$ of $\mathcal{N}$ where $t > d$. Delete from each $N_i \in \mathcal{N}$ all
    leaf labels $x_{d+1}, \dots, x_t$, suppress any resulting node of  degree $2$ and delete any resulting unlabelled leaves of degree 1.
\end{description}

\begin{lem}
\label{lem:utckernel}
There exists a kernelization for {\utc} producing an instance $(N'',T'')$ with at most $\max( 6k, 4 )$ taxa and $\max( 15k, 5)$ edges, where~$k = r(N'') \leq r(N)$.\end{lem}
\begin{proof}

If, during the kernelization procedure, we ever discover that the answer  to {\utc}  is definitely NO (respectively, YES) then we simply output a trivial NO (YES) instance as $(N'',T'')$ e.g.  letting $N''$ and $T''$ be two topologically distinct (identical) unrooted phylogenetic trees on 4 taxa and 5 edges. We shall henceforth use this implicitly; this is where the ``4'' and ``5'' terms come from in the statement of the lemma. Note that if $|X| \leq 3$, the answer is trivially YES, so we henceforth assume $|X| \geq 4$.

We begin with some trivial pre-processing. If $N$ contains a cut-edge $e$ such that one of the two connected components obtained by deleting $e$ contains no taxa, we delete $e$
and this component from $N$ and suppress the degree 2 node created by deletion of $e$. (This is safe, i.e. does not alter the YES/NO status of the answer to {\utc} because the
image of $T$ in $N$ can never enter such a component). We repeat this step until it no longer holds. Let $N'$ be the resulting network. If $N'$ and $T$ are both trees, and are
topologically distinct (respectively, identical) the answer is definitely NO (YES).  Hence, we assume that $N'$ is not a tree.

Next, we apply the \textbf{Common Pendant Subtree (CPS) reduction} to $\{N',T\}$ until it can no longer be applied.  It is easy to see that applying this reduction is safe.
This is because the image of the common pendant subtree, and the common pendant subtree itself, are necessarily identical in $N'$. Gently abusing notation, let $N'$ be the
resulting network and $T'$ the resulting tree. Observe that at this stage $N'$ potentially still contains pendant subtrees (with 2 or more taxa). This occurs if the pendant
subtree has no common counterpart in $T'$. However, if this happens the answer is definitely NO. Therefore, we can henceforth assume that $N'$ contains no pendant subtrees
(with 2 or more taxa).

The next step is to apply the \textbf{Common $3$-Chain ($3$-CC) reduction} repeatedly to $\{N',T'\}$ until it can no longer be applied. This has the effect of clipping all
common chains on 4 or more taxa to length 3. (The fact that we can clip common chains to constant length is the reason we obtain a linear kernel). Let $(N'', T'')$ be the
instance obtained after a single application of the common chain reduction rule. To establish correctness it is sufficient to show that $(N'',T'')$ is a YES instance if and
only if $(N',T')$ is a YES instance.

It is easy to see that if $(N', T')$ is a YES instance then $(N'', T'')$ is a YES instance.  This is because, if $N'$ contains an image of $T'$, then an image of $T''$ (in
$N''$) can be obtained from the image of $T'$ simply by disregarding the surplus taxa deleted from the chain.

The other direction is somewhat more subtle. Observe that, prior to the chain reduction, the common chain $C' = (x_1, x_2, \ldots, x_t)$, $t \geq 4$, was not pendant in $N'$
(because $N'$ contained no pendant subtrees). Hence, the clipped chain $C'' = (x_1, x_2, x_3)$  is not pendant in $N''$. Let $e_1, e_{12}, e_{23}, e_{3}$ be the 4 interior
edges of $N''$ shown in Figure \ref{fig:greatfigure1}.

\begin{figure}[H]\centering\begin{tikzpicture}[scale=1]

\tikzset{lijn/.style={ultra thick}}
\draw[lijn] (0,0) -- (4,0);
\draw (0.5,0.3) node {$e_1$};
\draw (1.5,0.3) node {$e_{12}$};
\draw (2.5,0.3) node {$e_{23}$};
\draw (3.5,0.3) node {$e_3$};
\draw[very thick, fill, radius=0.1] (1,0) circle;
\draw[very thick, fill, radius=0.1] (1,-1) circle;
\draw (1,-1.3) node {$x_1$};
\draw[lijn] (1,0) -- (1,-1);
\draw[very thick, fill, radius=0.1] (2,0) circle;
\draw[very thick, fill, radius=0.1] (2,-1) circle;
\draw (2,-1.3) node {$x_2$};
\draw[lijn] (2,0) -- (2,-1);
\draw[very thick, fill, radius=0.1] (3,0) circle;
\draw[very thick, fill, radius=0.1] (3,-1) circle;
\draw (3,-1.3) node {$x_3$};
\draw[lijn] (3,0) -- (3,-1);
\end{tikzpicture}
\caption{\emph{The chain $C''$ in $N''$.}} \label{fig:greatfigure1}
\end{figure}
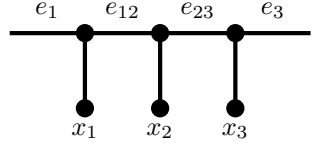

Now, suppose $N''$ displays $T''$; we will prove that $N'$ displays $T'$.  Fix some image of $T''$ inside $N''$. We distinguish two main cases. Note that $C'$ is pendant in $T'$ if and only if $C''$ is pendant in $T''$.\\
\\
\noindent \emph{Case 1: $C''$ is \emph{not} pendant in $T''$.} Both $e_1$ and $e_3$ must be on the image of $T''$ in $N''$, because otherwise the image of the chain $C''$ is
pendant, a contradiction. If both $e_{12}$ and $e_{23}$ are also on the image, then the chain $C''$ and its image in $N''$ are identical. In particular, there is no ambiguity
about the orientation of the chain, so reintroducing the clipped taxa $(x_4, \ldots, x_t)$ into the image of $T''$ (next to $x_3$) yields an image of $T'$ in $N'$. The only
remaining subcase is that, in addition to both $e_1$ and $e_3$, exactly one of $\{e_{12}, e_{23}\}$ is on the image. Without loss of generality let this be $e_{12}$. However,
this is not possible, because it would mean that $\{x_1, x_2\}$ are pendant in the image of $C''$, and
this cannot be an image of $T''$ because $\{x_1, x_2\}$ are not pendant in $T''$.\\
\\
\emph{Case 2: $C''$ \emph{is} pendant in $T''$.} There are two subcases to consider.
\begin{itemize}
\item In the first subcase, $x_1$ and $x_2$ share a parent in $T''$. (That is, the chain is oriented \emph{towards} the rest of the tree). In such a situation both $e_{12}$ and $e_{3}$ must be on the image of $T''$. (If this was not the
case, $\{x_2, x_3\}$ would be pendant in the image of $C''$, but this is not possible because they are not pendant in
$T''$.)  Now, if $e_{23}$ is on the image (irrespective of whether $e_1$ is on the image) then, as in the earlier
case, reintroducing the clipped taxa $(x_4, \ldots, x_t)$ into the image of $T''$ (next to $x_3$) yields an image of $T'$ in $N'$. The main subtlety is if $e_{23}$ is not on the image, and (necessarily) $e_1$ is. This occurs if the image of $C''$
exits via $e_1$, follows some simple path $P$ through another part of the network, and re-enters at $e_3$. However,
note that, within the image, the path $P$ contains exactly one node of degree 3 - which is the image of the parent of $x_3$ -  and for the rest only degree 2 nodes. This means that we can manipulate the image of $T''$ as follows: put $e_{23}$ in the image, remove $e_1$ from the image, and then tidy up the image in the usual sense (i.e. repeatedly deleting unlabelled nodes of degree 1). This is a new, valid image of $T''$, and puts us back in the situation when $e_{23}$ \emph{is} on the image, so we are done.
\item In the second subcase, $x_2$ and $x_3$ share a parent in $T''$. (That is, the chain is oriented \emph{away} from the rest of the tree). Observe that $e_1$ and $e_{23}$ must be on the image, because otherwise $\{x_1, x_2\}$ is pendant
in the chain image but not in $T''$. If $e_{12}$ is in the image (irrespective of whether $e_3$ is in the image), re-introducing the clipped taxa $(x_4, \ldots, x_t)$ to the right of $x_3$ yields an image of $T'$ in $N'$. Again, there is one subtle situation, and that is when $e_{12}$ is not on the image, but $e_3$ is. Just as before this occurs if the image of $C''$
exits via $e_1$, follows some simple path $P$ through another part of the network, and re-enters at $e_3$. The unique
node on $P$ of degree 3 is the image of the parent of $x_1$ (and all other nodes on $P$ are degree 2). Hence, if we put $e_{12}$ into the image, remove $e_3$ from the image, and tidy the image up, we obtain a new valid image of $T''$ and we are back in the case when $e_{12}$ is in the image.
\end{itemize}

Thus,  we have established that if $N''$ displays $T''$, then $N'$ displays $T'$. Hence, an application of the 3-CC chain reduction is always safe.


Although it can be shown that the CPS and 3-CC reduction rules  produce a linear-size kernel, we now add a third reduction rule which helps to further reduce the size of the
kernel and (more importantly) the value of the parameter (i.e., the reticulation number). Assume that none of the previous reduction rules are applicable.

\medskip
\noindent \textbf{Network chain (NC) reduction.}  If the network contains a chain $(x_1,\ldots ,x_t)$ with~$t\geq 3$ then proceed as follows. Let~$e_{i,i+1}$ be the edge
connecting the parents of~$x_i$ and~$x_{i+1}$.  Let~$e_1$ be the edge incident to the parent of~$x_1$ that is not~$e_{12}$ and not incident to~$x_1$. Let~$e_t$ be the edge
incident to the parent of~$x_t$ that is not~$e_{t-1,t}$ and not incident to~$x_t$. (Note that all these edges exist, because the network does not contain any pendant subtrees,
and thus no pendant chains.)

\begin{enumerate}
\item If $t\geq 7$ then return a trivial NO instance.
\item If $t=6$ then delete~$e_{34}$.
\item If $t=5$, do the following. If the tree contains a chain $(x_1,x_2,x_3)$, delete~$e_{34}$. Otherwise, delete~$e_{23}$.
\item If $t=4$, do the following. If the tree contains a chain  $(x_1,x_2,x_3)$, delete~$e_{34}$. If it contains a chain $(x_2,x_3,x_4)$, delete~$e_{12}$. Otherwise,
    delete~$e_{23}$.
\item If~$t=3$ and the tree has a pendant subtree on~$\{x_1,x_2,x_3\}$, do the following. If~$x_1$ and~$x_2$ have a common parent in the tree, delete~$e_1$. Otherwise,
    delete~$e_3$.
\item Otherwise, if~$t=3$ and the tree has a pendant subtree on~$\{x_1,x_2\}$, delete~$e_{23}$.
\item Otherwise, if~$t=3$ and the tree has a pendant subtree on~$\{x_2,x_3\}$, delete~$e_{12}$.
\item Otherwise, if~$t=3$ and the tree has a chain~$(x_1,x_2,x_3)$, delete~$x_3$.
\end{enumerate}
In all cases, we suppress any resulting degree-2 nodes.

We now show that the network chain reduction (\textbf{NC}) rule is safe.  Suppose that the network displays the tree. Then the chain $(x_1,\ldots ,x_t)$ of the network is
either also a chain of the tree, or there exists some~$1\leq i\leq t-1$ such that the tree has pendant chains on $\{x_1,\ldots ,x_i\}$ and on $\{x_{i+1},\ldots ,x_t\}$. We now
discuss each case of the network chain reduction separately.

\begin{enumerate}
\item In this case it follows that there is a common chain of length at least four, which is not possible since we assumed that the 3-CC common chain reduction is not applicable.
\item This is only possible if $(x_1,x_2,x_3)$ and $(x_4,x_5,x_6)$ are pendant chains of the tree.  Hence, $e_{34}$ is not used by any image of the tree in the network and
    can be deleted.
\item If the tree contains a chain $(x_1,x_2,x_3)$, then it must be pendant.  Hence,~$e_{34}$ can be deleted. Otherwise, $(x_3,x_4,x_5)$ must be a pendant chain of the tree
    and~$e_{23}$ can be deleted.
\item Similar to the previous case.  If neither $(x_1,x_2,x_3)$ nor $(x_2,x_3,x_4)$ is a pendant chain of the tree, then $(x_1,x_2)$ and $(x_3,x_4)$ must both be pendant
    chains of the tree, in which case $e_{23}$ can be deleted.
\item[5-7.] Similar to the previous cases.
\item[8.] In this case, $(x_1,x_2,x_3)$ is a chain of the tree that is not pendant (since otherwise we would be in one of the previous cases).  The image of the tree in the
    network must then use all of~$e_1,e_{12},e_{23},e_3$. Now we delete~$x_3$ and suppress the resulting degree-2 node. Hence the reduced network has a chain $(x_1,x_2)$ with
    edges~$e_1,e_{12},e_2$ defined as in the network chain reduction rule. To see that this reduction is safe, assume that the reduced tree is displayed by the reduced
    network. Then the embedding of the tree in the network has to use~$e_1$ and~$e_2$. It does not necessarily use~$e_{12}$ but if it does not it is easy to adapt the image
    such that it does use~$e_{12}$. Hence, the chain~$(x_1,x_2)$ can be replaced by $(x_1,x_2,x_3)$ and it follows that the original tree is displayed by the original
    network.
\end{enumerate}

Let $(N'',T'')$ be an instance obtained by applying the CPS, 3-CC and NC reduction rules exhaustively until none applies. Clearly, the process by which $(N'',T'')$ is obtained
from the original $(N,T)$ can be completed in polynomial time, since all pre-processing steps delete at least one node or edge from the network. It is easy to verify that, by
construction, $r(N'') \leq r(N)$. Hence, to complete the kernelization it remains only to show that the size of the instance $(N'',T'')$ is at most a linear function of
$r(N'')$, where for brevity we let $k = r(N'')$.  To see this, recall firstly that $N''$ has no pendant subtrees. Let $N'' = (V'',E'')$. Suppose we delete all taxa in $N''$ and
then repeatedly suppress nodes of degree 2, and delete nodes of degree 1, until neither of these operations can be applied anymore. For $k\geq 2$, this creates a $3$-regular
graph $N^{*}$ with nodes $V^{*}$ and edges $E^{*}$, that potentially contains multi-edges and loops. Notice that in each deletion of a leaf and each suppression of a node with
degree 2, we diminish both the number of nodes and the number of edges by 1. Since we started out with $|E''|=k+|V''|-1$ we still have $|E^{*}|=k+|V^{*}|-1$. Moreover, because
of $3$-regularity, $|E^{*}|=3|V^{*}|/2$. Combining yields $|V^{*}|=2k-2$ and therefore $|E^{*}|=3(k-1)$. (For $k=1$, $N^{*}$ contains no nodes and is strictly speaking not a
graph: in this case we define $N^{*}$ to be a single node with a loop). Observe that $N''$ can be obtained from $N^{*}$ by replacing each edge with a chain of taxa: this
operation is sufficient because $N''$ had no pendant subtrees. Moreover, each such chain contains at most two leaves since otherwise the network chain reduction rule would be
applicable. This means that $|X''|$ is at most $2 \cdot \max(1, 3(k-1))$, and the number of edges in $N''$ is at most $5 \cdot \max(1, 3(k-1))$.
\end{proof}


We observe that simply reducing common chains to length 2, i.e. applying the 2-CC reduction rule,  is \emph{not} safe, as the following example shows. Suppose $N$ consists of a
single cycle with taxa $a,b,c,d,e,f$ in that (circular) order. Let $T$ be a caterpillar tree with taxa $a,b,c,f,e,d$ in that order. $N$ does not display $T$. However, if the
common chain $(a,b,c)$ is clipped to $(a,b)$ - or to $(b,c)$ - the resulting network $N''$ \emph{does} display $T''$. A symmetrical argument also holds for the common chain
$(f,e,d)$.

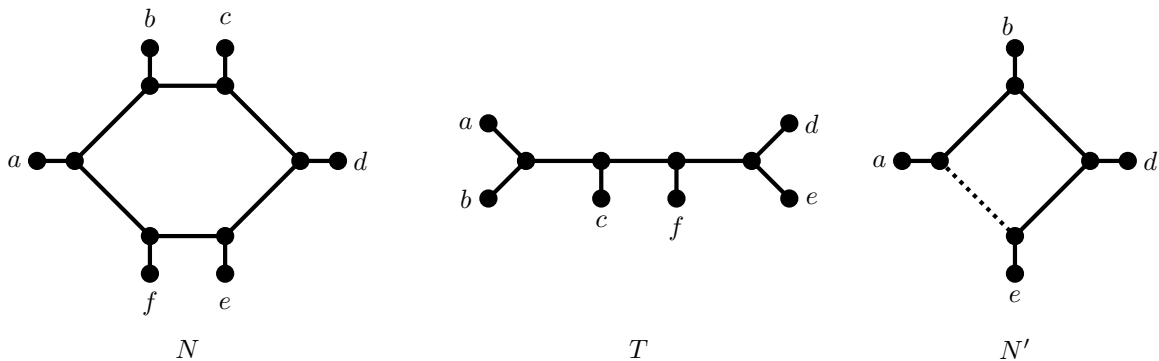
\begin{figure}[H]\centering\begin{tikzpicture}[scale=1]
\tikzset{lijn/.style={ultra thick}}
\tikzset{stippellijn/.style={ultra thick, dotted}}
\draw[very thick, fill, radius=0.1] (1,0) circle;
\draw[very thick, fill, radius=0.1] (2,0) circle;
\draw[very thick, fill, radius=0.1] (0,1) circle;
\draw[very thick, fill, radius=0.1] (1,2) circle;
\draw[very thick, fill, radius=0.1] (2,2) circle;
\draw[very thick, fill, radius=0.1] (3,1) circle;
\draw[lijn] (1,0) -- (2,0);
\draw[lijn] (2,0) -- (3,1);
\draw[lijn] (3,1) -- (2,2);
\draw[lijn] (2,2) -- (1,2);
\draw[lijn] (1,2) -- (0,1);
\draw[lijn] (0,1) -- (1,0);
\draw[very thick, fill, radius=0.1] (-.5,1) circle;
\draw[lijn] (-.5,1) -- (0,1);
\draw (-.8,1) node {$a$};
\draw[very thick, fill, radius=0.1] (1,2.5) circle;
\draw[lijn] (1,2.5) -- (1,2);
\draw (1,2.9) node {$b$};
\draw[very thick, fill, radius=0.1] (2,2.5) circle;
\draw[lijn] (2,2.5) -- (2,2);
\draw (2,2.9) node {$c$};
\draw[very thick, fill, radius=0.1] (1,-.5) circle;
\draw[lijn] (1,-.5) -- (1,0);
\draw (1,-.9) node {$f$};
\draw[very thick, fill, radius=0.1] (2,-.5) circle;
\draw[lijn] (2,-.5) -- (2,0);
\draw (2,-.9) node {$e$};
\draw[very thick, fill, radius=0.1] (3.5,1) circle;
\draw[lijn] (3.5,1) -- (3,1);
\draw (3.8,1) node {$d$};
\draw (1.5,-1.5) node {$N$};
\draw[very thick, fill, radius=0.1] (5.5,1.5) circle;
\draw[lijn] (5.5,1.5) -- (6,1);
\draw (5.2,1.5) node {$a$};
\draw[very thick, fill, radius=0.1] (5.5,0.5) circle;
\draw[lijn] (5.5,0.5) -- (6,1);
\draw (5.2,0.5) node {$b$};
\draw[very thick, fill, radius=0.1] (6,1) circle;
\draw[lijn] (6,1) -- (7,1);
\draw[very thick, fill, radius=0.1] (7,1) circle;
\draw[lijn] (7,1) -- (8,1);
\draw[very thick, fill, radius=0.1] (8,1) circle;
\draw[lijn] (8,1) -- (9,1);
\draw[very thick, fill, radius=0.1] (9,1) circle;
\draw[very thick, fill, radius=0.1] (9.5,0.5) circle;
\draw (9.8,0.5) node {$e$};
\draw[lijn] (9,1) -- (9.5,0.5);
\draw[very thick, fill, radius=0.1] (9.5,1.5) circle;
\draw[lijn] (9,1) -- (9.5,1.5);
\draw (9.8,1.5) node {$d$};
\draw[very thick, fill, radius=0.1] (7,.5) circle;
\draw[lijn] (7,1) -- (7,.5);
\draw (7,.2) node {$c$};
\draw[very thick, fill, radius=0.1] (8,.5) circle;
\draw[lijn] (8,1) -- (8,.5);
\draw (8,.1) node {$f$};
\draw (7.5,-1.5) node {$T$};
\draw[very thick, fill, radius=0.1] (11.5,1) circle;
\draw[very thick, fill, radius=0.1] (11,1) circle;
\draw[lijn] (11.5,1) -- (11,1);
\draw (10.7,1) node {$a$};
\draw[very thick, fill, radius=0.1] (12.5,2) circle;
\draw[very thick, fill, radius=0.1] (12.5,2.5) circle;
\draw[lijn] (12.5,2) -- (12.5,2.5);
\draw (12.4,2.8) node {$b$};
\draw[lijn] (11.5,1) -- (12.5,2);
\draw[very thick, fill, radius=0.1] (13.5,1) circle;
\draw[very thick, fill, radius=0.1] (14,1) circle;
\draw[lijn] (13.5,1) -- (14,1);
\draw (14.3,1) node {$d$};
\draw[lijn] (12.5,2) -- (13.5,1);
\draw[very thick, fill, radius=0.1] (12.5,0) circle;
\draw[very thick, fill, radius=0.1] (12.5,-.5) circle;
\draw[lijn] (12.5,0) -- (12.5,-.5);
\draw (12.5,-.8) node {$e$};
\draw[lijn] (13.5,1) -- (12.5,0);
\draw[stippellijn] (11.5,1) -- (12.5,0);
\draw (12.5,-1.5) node {$N'$};
\end{tikzpicture}
\caption{\emph{Example showing that it is not safe to reduce chains to length~2. The shown network~$N$  does not display the given tree~$T$. However, if the chains $(a,b,c)$ and $(d,e,f)$ are reduced to length~2, then the reduced network~$N'$ does display the reduced tree~$T'$ (by deleting the dotted edge).}}
\label{fig:greatfigure}\end{figure}

The proof of the  FPT result follows by applying a simple bounded-search branching algorithm to the kernelized instance. Note that, as mentioned earlier, this algorithm can be
applied independently of the kernelization.

\begin{thm}\label{thm:UTCFPT}
Let $(N,T)$ be an input to \textsc{UTC}, where $N=(V,E)$. There exists an $O(4^kn^2)$-time algorithm for UTC, where~$k= r(N)$ and~$n = |V|$.\end{thm}
\begin{proof}
If the network is a tree then the problem can be solved easily in polynomial time by deciding whether or not the network is isomorphic to the input tree. Otherwise, we proceed as follows.

Consider any two taxa~$x,y$ that have a common neighbour in the tree~$T$. If~$x$ and~$y$ also have a common neighbour in~$N$, then we can delete~$y$ from both~$T$ and~$N$ and suppress the resulting degree-2 nodes (see the CPS reduction above).

Otherwise, let~$n_x$ and~$n_y$ be the neighbours of, respectively,~$x$ and~$y$ in the network~$N$. Let~$e_1,e_2$ be the two edges that are incident to~$n_x$ but not to~$x$ and let~$e_3,e_4$ be the two edges that are incident to~$n_y$ but not to~$y$. If~$N$ displays~$T$, then the embedding of~$T$ in~$N$ can contain at most three of these four edges~$e_1,\ldots ,e_4$ (since there is exactly one edge leaving the path between~$x$ and~$y$ in the embedding). Hence, we create four subproblems~$P_1,\ldots ,P_4$. In subproblem~$P_i$, we delete edge~$e_i$ and suppress resulting degree-2 nodes. The parameter (reticulation number) in each subproblem is~$k-1$. Hence, the running time is $O(4^kn^2)$.
\end{proof}

\section{Unrooted hybridization number (UHN) on two trees}
\label{sec:uhn}

In this section we study the unrooted hybridization number problem in case the input consists of two trees $T_1, T_2$ and we show equivalence to a well-known problem that has
been studied before in the literature, namely the \emph{Tree Bisection and Reconnect} problem.

Let $T$ be an unrooted, binary tree on $X$.  A \emph{Tree Bisection and Reconnect} (TBR) move is defined as follows: (1) we delete an edge of $T$ to obtain a forest consisting
of two subtrees $T'$ and $T''$. (2) Then we select two  edges $e_1 \in T', e_2 \in T''$, subdivide these two edges with two new nodes $v_1$ and $v_2$, add an edge from $v_1$ to
$v_2$, and finally suppress all nodes of degree 2. In case either $T'$ or $T''$ are single leaves, then the new edge connecting $T'$ and $T''$ is incident to that node. Let
$T_1, T_2$ be two binary and unrooted trees on the same set of leaf-labels. The TBR-distance from $T_1$ to $T_2$, denoted $d_{TBR}(T_1, T_2)$, is simply the \textit{minimum}
number of TBR moves required to transform $T_1$ into $T_2$.

It is well known that computation of TBR-distance is essentially equivalent to the \textsc{Maximum Agreement Forest (MAF)} problem, which we now define. Given an unrooted,
binary tree on $X$ and $X' \subset X$ we let $T(X')$ denote the minimal subtree that connects all the elements in $X'$. An \emph{agreement forest} of two unrooted binary trees
$T_1, T_2$ on $X$ is a partition of $X$ into non-empty blocks $\{X_1, \ldots, X_k\}$ such that (1) for each $i \neq j$, $T_1(X_i)$ and $T_1(X_j)$ are node-disjoint and
$T_2(X_i)$ and $T_2(X_j)$ are node-disjoint, (2) for each $i$, $T_1|X_i = T_2|X_i$. A \emph{maximum agreement forest} is an agreement forest with a minimum number of
components, and this minimum is denoted $d_{MAF}(T_1,T_2)$. In 2001 it was proven by Allen and Steel that $d_{MAF}(T_1, T_2) = d_{TBR}(T_1, T_2) + 1$ \cite{AllenSteel2001}.

\begin{thm}
Let $T_1, T_2$ be two unrooted binary phylogenetic trees on the same set of taxa $X$. Then $d_{TBR}(T_1, T_2) = h^u(T_1,T_2)$. \label{thm:tbrIsMaf}
\end{thm}
\begin{proof}
We first show  $h^u(T_1,T_2) \leq d_{TBR}(T_1,T_2)$. Let $d_{TBR}(T_1,T_2)=k$. Observe that if $k=0$ then $T_1 = T_2$, because $d_{TBR}$ is a metric, and if $T_1 = T_2$ then $h^{u}(T_1, T_2)=0$, so the claim holds. Hence, assume $k \geq 1$.

By the earlier discussed equivalence, $T_1$ and $T_2$ have an agreement forest with $k+1$ components $F = \{F_0, \ldots, F_{k}\}$. Our basic strategy is to start with a network
that trivially displays $T_1$  (specifically, $T_1$ itself) and then to ``wire together'' the components of $F$ such that an image of $T_2$ is progressively grown. Each such
wiring step involves subdividing two edges and introducing a new edge between the two subdivision nodes. This increases the number of nodes in the network by 2 and the number
of edges by 3, so it increases the reticulation number by 1. We will do this $k$ times, yielding the desired result.

Observe that for least one of the components, $F_{p}$ say, $T_2(F_p)$ will be pendant in $T_2$.  Let $F' = F \setminus \{F_p\}$, $X' = X \setminus F_{p}$, $T'_1 = T_1|X'$ and
$T'_2 = T_2|X'$. Let $\{u,v\}$ be the edge that, when deleted, detaches $T_2(F_{p})$ from the rest of the tree. Assume without loss of generality that $u$ lies on $T_2(F_{p})$
and $v$ lies on $T_2( X' )$. The nodes $u$ and $v$ thus lie on unique edges of $T_2 | F_p$ and $T_2 | X'$ (or taxa if $F_p$ and/or $X'$ are singleton sets); these can be viewed
as the wiring points where $F_p$ wants to connect to the rest of the tree. Next, observe that $F'$ is an agreement forest for $T'_1$ and $T'_2$, so it too has a pendant
component, and the process can thus be iterated. In this way we can impose an elimination ordering on the components of $F$. For the sake of brevity assume that the components
$F_0, F_1, \ldots, F_k$ are already ordered in this way.

Now, set $N_k := T_1$.  For each $F_i \in F$, fix the unique image of $F_i$ in $N_k$ (this allows us without ambiguity to refer to \emph{the} image of $F_i$ in the intermediate
networks we construct). For each $ 0 \leq j \leq k-1$, we construct $N_{j}$ from $N_{j+1}$ in the following way. Assume that by construction $N_{j+1}$ already contains an image
of $T_2 | (\cup_{j'  > j} F_{j'})$ and an image of $T_2|F_j$, and that these images are disjoint. (Clearly this is true for $j=k-1$, by the definition of agreement forest).
From the earlier argument we know the two wiring points at which $T_2 | F_j$ wishes to join with $T_2 | (\cup_{j'  > j} F_{j'})$. If $|F_j| \geq 2$ the wiring point within
$F_j$ will be an edge, otherwise it is a taxon, and an identical statement holds for  $|\cup_{j' >j} F_{j'}|$. Assume for now that both wiring points are edges, $e_1$ and $e_2$
respectively. The images of these edges will, in general, be paths in $N_{j+1}$. We subdivide any edge on the image of $e_1$, and any edge on the image of $e_2$, and connect
them by a new edge. If a wiring point is a taxon $x$ the only difference is that we subdivide the unique edge entering $x$ in $N_{j-1}$. At the end of this process, $N_0$
displays both $T_1$ and $T_2$. This completes the claim $h^{u}(T_1,T_2) \leq d_{TBR}(T_1,T_2)$.

To prove $h^{u}(T_1,T_2) \geq d_{TBR}(T_1,T_2)$, let $k = h^u(T_1,T_2)$ and let $N$ be an unrooted phylogenetic network with reticulation number $k$ that displays both $T_1$
and $T_2$. Fix an image $T'_1$ of $T_1$ inside $N$. If this image is not a spanning tree of $N$, greedily add edges to the image until it becomes one. (The edges added this way
will correspond to unlabelled degree 1 nodes that are repeatedly deleted when tidying up the image to obtain $T_1$). Now, fix an image $T'_2$ of $T_2$ inside $N$.  Let
$F\subseteq E(N)$ be those edges of $N$ that are \emph{only} in $T'_2$. Deleting in $T_2$ the edges that correspond to $F$ breaks $T_2$ up into a forest with at most $|F|+1$
components. In fact, by construction this will be an agreement forest. Hence, $d_{TBR}(T_1, T_2) \leq |F|$. What remains is to show that $|F| \leq h^{u}(T_1,T_2)$. Given that
$T'_1$ was a spanning tree of $N$, and none of the edges on this image are in $F$, the graph $(V, E-F)$ is connected, so $|E|-|F| \geq |V|-1$. Hence, $|F| \leq |E|-|V|+1 = k$.
\end{proof}

\noindent Note that the proof given above is constructive, in the following sense. Given an agreement forest $F$ with $k+1$ components, one can easily construct in polynomial
time an unrooted network $N$ with reticulation number $k$ that displays both the trees, and given an unrooted network $N$ with reticulation number $k$ (and images of $T_1$ and
$T_2$ in $N$) one can easily construct in polynomial time an agreement forest $F$ with $k+1$ components.

\begin{cor}
\textsc{UHN} is NP-hard, in APX, and FPT in parameter $h^{u}(T_1, T_2)$.
\end{cor}
\begin{proof}
Immediate from Theorem \ref{thm:tbrIsMaf} and the corresponding results for $d_{TBR}$. Hardness (and a linear-size kernel) were established in \cite{AllenSteel2001}. The
best-known approximation result for $d_{TBR}$ is currently a polynomial-time 3-approximation \cite{whidden2013fixed,whiddenWABI}. The best-known FPT result for $d_{TBR}$
is currently $O( 3^{k} \cdot \text{poly}(n))$ \cite{chen2015parameterized}.
\end{proof}

\section{Root-uncertain hybridization number (RUHN)}
\label{sec:ruhn}

In this section we turn our attention to the \textsc{Root Uncertain Hybridization Number} (\ruhn) problem. We remind the reader that in this problem  the input consists of a
set of \emph{unrooted} binary trees and we are asked to choose the root location of each tree, such that the hybridization number is minimized. In the first part of this
section we show that {\ruhn} is already NP-hard and APX-hard even when the input consists of two trees. On the other side, in the next subsection
we show that the problem is FPT in the hybridization number for any number of trees by providing a quadratic-sized kernel.  We conclude the section by discussing how an
exponential-time algorithm can be obtained for solving the kernel.

\subsection{Hardness}
\begin{lem}
\label{lem:similar}
Let $\mathcal{T} = \{ T_1$, $T_2 \}$ be an input to {\hn}. We can transform in polynomial time $T_1$ and $T_2$ into two unrooted binary phylogenetic trees $T^{*}_1$, $T^{*}_2$ such that,
\begin{equation}
\label{eq:plus1}
h^{ru}( T_1^{*}, T_2^{*} ) = h^{r}( T_1, T_2) + 1.
\end{equation}
\end{lem}
\begin{proof}
Let $X$ denote the taxa of $T_1$ and $T_2$ and let $n = |X|$. We will construct in polynomial time a pair of unrooted trees $T_1^{*}, T_2^{*}$ on $3|X|+2$ taxa such that
(\ref{eq:plus1}) holds.

To construct $T^{*}_1$, we start by taking an unrooted caterpillar $(c_{0}, c_{1}, ..., c_n, d_0, d_1, d_2, ..., d_{n})$ on $2n+2$ new taxa. Let $r_1$ be the root of $T_1$. To
complete $T_1^{*}$ we ignore all the directions on the arcs of $T_1$, and concatenate the caterpillar to $T_1$ by subdividing the unique edge entering $d_n$ with a new node
$u$, and connect $u$ to $r_1$. The construction of $T^{*}_2$ is analogous, except that the $c$-part of the caterpillar is reversed: $(c_n, c_{n-1}, ..., c_0, d_0, d_1, d_2...,
d_n)$. See Figure \ref{fig:solution} (left and centre) for an example when $n=5$.

\begin{figure}[H]\centering\begin{tikzpicture}[scale=0.6]
\draw[very thick] (0,1) -- ++(-1.5,-1.5) -- ++(3,0) -- ++(-1.5,1.5);
\draw[very thick, fill, radius=0.1] (0,1) circle;
\node at (0,0) {$T_1$};
\node at (0,-1.2) {$T^{*}_1$};
\foreach \n in {5,4,3,2,1,0} \draw[very thick, fill, radius=0.1] (0,6-\n) -- ++(0,1) circle -- ++(1,0) circle node[right] {$d_\n$};
\foreach \n in {1,2,3,4,5} \draw[very thick, fill, radius=0.1] (0,12-\n) -- ++(0,1) circle -- ++(1,0) circle node[right] {$c_\n$};
\draw[very thick, fill, radius=0.1] (0,12) -- ++(0,1) circle node[above] {$c_0$};
\draw[<-, very thick, fill, radius=0.1] (-0.1,10.5) -- ++(-1.4,0);
\draw[very thick, fill, radius=0.1] (-1.5,10.5) circle node[above] {\text{root}};

\draw[very thick] (6,1) -- ++(-1.5,-1.5) -- ++(3,0) -- ++(-1.5,1.5); \draw[very thick, fill, radius=0.1] (6,1) circle; \node at (6,0) {$T_2$}; \node at (6,-1.2) {$T^{*}_2$};
\foreach \n in {5,4,3,2,1,0} \draw[very thick, fill, radius=0.1] (6,6-\n) -- ++(0,1) circle -- ++(1,0) circle node[right] {$d_\n$}; \foreach \n in {0,1,2,3,4} \draw[very thick,
fill, radius=0.1] (6,7+\n) -- ++(0,1) circle -- ++(-1,0) circle node[left]  {$c_\n$}; \draw[very thick, fill, radius=0.1] (6,12) -- ++(0,1) circle node[above] {$c_5$};
\draw[->, very thick, fill, radius=0.1] (7.5,10.5) -- ++(-1.4,0); \draw[very thick, fill, radius=0.1] ( 7.5,10.5) circle node[above] {\text{root}};

\draw[very thick] (12,1) -- ++(-1.5,-1.5) -- ++(3,0) -- ++(-1.5,1.5); \draw[very thick, fill, radius=0.1] (12,1) circle; \node at (12,0) {$N$}; \node at (12,-1.2) {$N'$};
\foreach \n in {5,4,3,2,1,0} \draw[very thick, fill, radius=0.1] (12,6-\n) -- ++(0,1) circle -- ++(1,0) circle node[right] {$d_\n$}; \draw[very thick, fill, radius=0.1] (12,7)
-- ++(0,1) circle; \draw[very thick, fill, radius=0.1] (12,8) -- ++(-1,1) circle -- ++(-1,0) circle node[left]  {$c_0$}; \draw[very thick, fill, radius=0.1] (11,9) -- ++(
0,1.0) circle -- ++(-1,0) circle node[left]  {$c_1$}; \draw[very thick, fill, radius=0.1] (11,10) -- ++( 0,1.0) circle -- ++(-1,0) circle node[left]  {$c_2$}; \draw[very thick,
fill, radius=0.1] (12,8) -- ++( 1,1) circle -- ++( 1,0) circle node[right] {$c_5$}; \draw[very thick, fill, radius=0.1] (13,9) -- ++( 0,1.0) circle -- ++( 1,0) circle
node[right] {$c_4$}; \draw[very thick, fill, radius=0.1] (13,10) -- ++( 0,1.0) circle -- ++( 1,0) circle node[right] {$c_3$}; \draw[very thick, fill, radius=0.1] (11,11) --
++(1,1) circle -- ++(1,-1);

\end{tikzpicture}
\caption{\emph{An example of the transformation used in Lemma \ref{lem:similar} when $|X|=5$. Left and centre: the two unrooted binary trees $T^{*}_1$ and $T^{*}_2$ that are
used as input to {\ruhn}. These are obtained from the original rooted binary trees $T_1$ and $T_2$ on $X$ that are the input to the {\hn} problem. If these trees are rooted at
the specified points, then the rooted phylogenetic network $N'$ displays the two rootings, where $N$ is an optimal solution to the original $HN$ problem. (Although not shown
explicitly here, in the top part of $N'$ all arcs are oriented downwards.)}} \label{fig:solution}\end{figure}
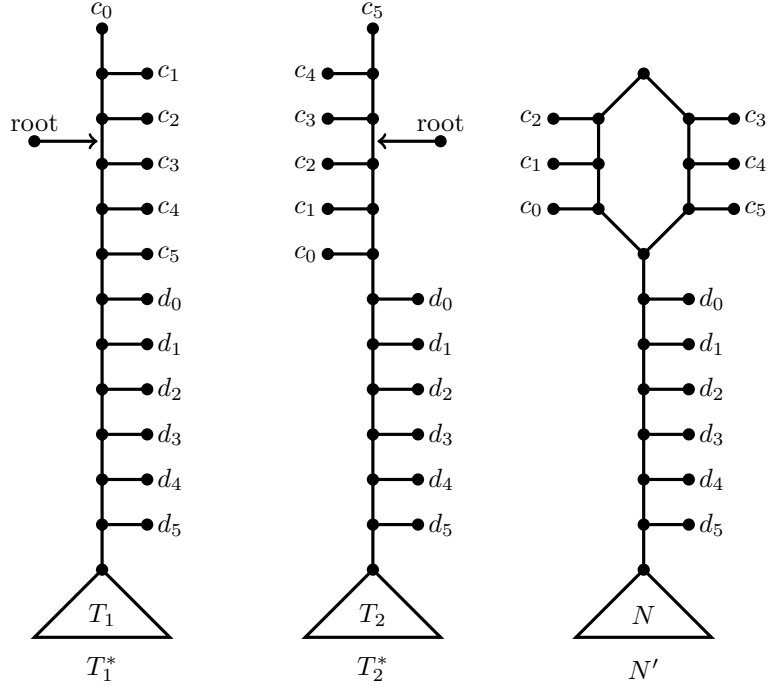

It is quite easy to show that  $h^{ru}( T_1^{*}, T_2^{*} ) \leq h^{r}( T_1, T_2) + 1$. Specifically, let $N$ be any optimum solution to the original {\hn} problem, i.e. $r(N) =
h^{r}(T_1, T_2)$. If we root both $T^{*}_1$ and $T^{*}_2$ on the internal edge between $c_{2}$ and $c_3$, then the network $N'$ as shown in Figure \ref{fig:solution} (right)
clearly displays these two rootings. Essentially, $N'$ has been obtained by  adding a single ``root cycle'' above $N$, so $r(N') = r(N)+1$. More formally, in order of
increasing distance from the root, the network $N'$ has taxa $c_2, c_1, c_0$ on one side of the root cycle, and $c_{3}, ..., c_{n-1}, c_{n}$ on the other.

The lower bound, $h^{ru}(T^{*}_1, T^{*}_2) \geq h^{r}(T_1, T_2) + 1$, requires slightly more effort to prove. We will use the following observation.
\begin{equation}
\label{eq:ub}
h^{ru}(T^{*}_1, T^{*}_2) \leq h^{r}(T_1, T_2) + 1 \leq (n-2) + 1 = n-1.
\end{equation}
The second inequality follows from the well-known fact that two rooted binary phylogenetic trees on $n \geq 2$ taxa can have hybridization number at most $n-2$ \cite{BaroniEtAl2005}.

Notice that, if in a rooting of $T^{*}_1$,  the whole $c$-part of the caterpillar appears in reverse order of the one in a rooting of $T^{*}_2$ then just this $c$-part of the
caterpillars adds $n-1$ to the hybridization number of that rooting. The same holds for the $d$-parts of the caterpillars. In both cases, using the observation above, the lower
bound is true. In particular, this implies that the lower bound holds if  $T^{*}_1$ is rooted inside its $T_1$-part, since any rooting of $T^{*}_2$ will create either
oppositely-oriented $c$- or $d$-parts of the caterpillars. The same holds for a rooting inside $T_2$. But clearly, if both $T^{*}_1$ and $T^{*}_2$ are rooted outside their
$T$-parts, then these $T$-parts add $h^{r}(T_1, T_2)$ to the hybridization number of such a rooting. Since the caterpillars of $T^{*}_1$ and $T^{*}_2$ are non-isomorphic, any
rooting within the $c$- or $d$-parts of caterpillars of the trees will additionally add at least 1 to the hybridization number. (Formally speaking this last argument is a
consequence of the \emph{cluster reduction} described in \cite{bordewich2}).
\end{proof}

\begin{thm}
\label{thm:apxhard}
\textsc{Root Uncertain Minimimum Hybridization} is NP-hard and APX-hard for $|\mathcal{T}|=2$.
\end{thm}
\begin{proof}
{\hn} is already known to be NP-hard and APX-hard for $|\mathcal{T}|=2$. NP-hardness of {\ruhn} is thus immediate from Lemma \ref{lem:similar}. We can also use this lemma to
prove APX-hardness, which excludes the existence of a PTAS for {\ruhn}, unless P=NP. The APX-hardness result might seem intuitively obvious, since the $+1$ term in
(\ref{eq:plus1}) is of vanishing significance as $h^{r}(T_1, T_2)$ grows. However, there are quite some technicalities involved in the extraction of a solution for {\hn} from a
solution for {\ruhn}.  In particular, additional combinatorial insight is required. We give a $(2,1)$ L-reduction from {\hn} to {\ruhn}. In fact, this can be extended to an
$(\alpha,1)$ L-reduction for each $1 < \alpha < 2$. To avoid disrupting the flow of the paper we have deferred the details of the L-reduction to the appendix.
\end{proof}

Note that one consequence of the L-reduction given in the proof of Theorem \ref{thm:apxhard}  is that if {\ruhn} has a constant-factor polynomial-time approximation algorithm
(i.e. is in APX), then so does {\hn}. In \cite{approximationHN} it is proven that, if {\hn} is in APX, so is \textsc{Directed Feedback Vertex Set}. Hence the following
corollary is obtained.

\begin{cor}
If {\ruhn} is in APX, then so is {\hn} and thus also \textsc{Directed Feedback Vertex Set}.
\end{cor}

Determining whether \textsc{Directed Feedback Vertex Set} is in APX is a longstanding open problem in computer science; the corollary can thus be viewed
\emph{provisionally} as a strengthening of Theorem \ref{thm:apxhard}.

\subsection{Parameterized complexity of {\ruhn}}

In this subsection we will show that {\ruhn} is FPT when the parameter is $h^{ru}(\mathcal{T})$ (or, in other words, the size $k$ of the optimal solution for \ruhn). To prove
this, we will provide a kernel of quadratic size which (when combined with any exponential-time algorithm) will let us answer the question ``\textit{Is the optimal solution to \ruhn $\leq k$?}" in time $O(f(k) \cdot
\text{poly}(n))$ for some computable function $f$ that depends only on $k$.

For the kernelization proof we use the same ingredients introduced in Section \ref{subsec:utcfpt} and in particular the two reductions rules introduced there:  \textbf{Common
Pendant Subtree} (CPS) reduction and \textbf{Common $d$-Chain} ($d$-CC) reduction rules. We use them slightly differently from how they were defined there, because here the
input to each reduction rule is a set of unrooted binary trees, and within the common chain reduction we will take $d = 5k$ (i.e., long common chains will be truncated down to
length $5k$). In \cite{ierselLinz2013} the authors described how these two reduction rules can be used in the \textit{rooted} HN problem to reduce the initial instance
${\mathcal{T}}$ to a new kernelized instance of rooted binary phylogenetic trees $\mathcal{T}'$ on a set of leaf labels $X'$ such that $h^r(\mathcal{T}) \leq k \Leftrightarrow
h^r(\mathcal{T}') \leq k$ and, moreover, $|X'| = O(k^2)$. Here, we adapt their arguments to work for the unrooted case as well. Although this might seem a direct
generalization, additional technicalities must be addressed arising in root placement on the unrooted trees/networks.


We start by defining the concept of \emph{generator} \cite{kelk2014constructing} which will be used in the rest of the section: An \emph{$r$-reticulation generator} (for short,
$r$-generator) is defined to be a directed acyclic multigraph with a single node of indegree 0 and outdegree 1, precisely $r$ reticulation nodes (indegree 2 and outdegree at
most 1), and apart from that only nodes of indegree 1 and outdegree 2. The \emph{sides} of an $r$-generator are defined as the union of its edges (the edge sides) and its nodes
of indegree-2 and outdegree-0 (the node sides). Adding a set of labels $L$ to an edge side $(u, v)$ of an $r$-generator involves subdividing $(u, v)$ to a path of $|L|$
internal nodes and, for each such internal node $w$, adding a new leaf $w'$, an edge $(w, w')$, and labeling $w'$ with some taxon from $L$ (such that $L$ bijectively labels the
new leaves). On the other hand, adding a label $l$ to a node side $v$ consists of adding a new leaf $y$, an edge $(v, y)$ and labeling $y$ with $l$. In \cite{ierselLinz2013} it
was shown that if $G$ is an $r$-generator, then $G$ has at most $4r-1$ edge sides and at most $r$ node sides.

\begin{thm}
\label{thm:unrooted}
Let $\mathcal{T}$ be a collection of binary, unrooted, phylogenetic trees on a common set of leaf labels (taxa) $X$. Let $\mathcal{T}'$ be the set of
binary, unrooted, phylogenetic trees on $X'$ after we have applied the common pendant subtree (CPS) and the Common $5k$-Chain reduction rules, until no such rule can be
performed anymore.
Then $h^{ru}(\mathcal{T}) \leq k \Leftrightarrow h^{ru}(\mathcal{T'}) \leq k$ and, moreover, $|X'| = O(k^2)$.
\end{thm}

We will start by showing that the CPS reduction rule leaves the hybridization number unchanged:

\begin{claim}
Let $\mathcal{T}$ be a set of unrooted binary trees with leaves labeled bijectively by $X$. Let $T$ be a maximal common pendant subtree of $\mathcal{T}$ and let $\mathcal{T}'$
be the set of all trees in $\mathcal{T}$ after the application of the Common Pendant Subtree Reduction rule to~$T$. Then $h^{ru}(\mathcal{T}) \leq k \Leftrightarrow
h^{ru}(\mathcal{T'}) \leq k$.
\end{claim}

\begin{proof}
Let $N$ be the optimal (with the minimum reticulation number) network that displays the optimally rooted version of $\mathcal{T}$ and let $N'$ be the optimal network that
displays the optimally rooted reduced instance $\mathcal{T}'$ (after a single application of the CPS reduction rule).

($\Leftarrow$) Let $h^{ru}(\mathcal{T}') = r(N') = k$.
From $N'$ we will construct a rooted network $N$ with $k$ reticulation nodes that displays $\mathcal{T}$. Since $N'$ displays $\mathcal{T}'$ which is a collection of trees with
leaves bijectively labeled from $\{X \setminus \{ X_T\}\} \cup \{x\}$ (where, as before, $X_T$ is the set of taxa of $T$), simply replace on $N'$ the leaf $x$ with the common
pendant subtree $T$. We have a new network $N''$ whose reticulation number obviously is $k$ (we did not add/create any new reticulation nodes). The leaves of $N''$ are labeled
from $X$ (without $x$). It remains only to show that $N''$ displays $\mathcal{T}$ which is immediate since $T$ displays itself. Observe that the root placement on each tree $T
\in \mathcal{T}'$ is irrelevant.

($\Rightarrow$) For the other direction, consider $\mathcal{T}$ and let $\mathcal{T}^{\rho}$ be a rooting of all trees such that $h^{r}(\mathcal{T}^{\rho})$ is minimized. Let
$N$ be the rooted network displaying the trees in $\mathcal{T}^{\rho}$ and let $\rho(T)$, for $T \in \mathcal{T}$ be the actual root of $T$ (given by $\mathcal{T}^{\rho}$).
Similar for $N$. Let $T$ be the CPS of each member of $\mathcal{T}$. From $N$ we need to construct a new network $N'$ with $k$ reticulation nodes that displays all the trees in
$\mathcal{T}'$. The problem will be: what if $\exists T \in \mathcal{T}$ such that its root is inside $T$? In such cases, the CPS reduction rule will cut-off the root of this
tree and this will ``force" us to root $T$ in another location unaffected by the CPS reduction rule, which will potentially change the hybridization number of the resulting
instance. Given a rooting of all members of $\mathcal{T}$ and $N$ with $h^{ru}(\mathcal{T}) = k$, consider the following rootings for each $T' \in \mathcal{T}'$:  if $\rho(T)
\in T$ then root $T'$ (after the clipping of $T$) on the parent of $x$ (the new taxon replacing $T$). Else, leave the rooting unchanged. Now, from $N$, we need to create a new
rooted network $N'$ that displays $\mathcal{T}'$ such that its reticulation number is (not greater than) $k$. Apply the standard procedure: let $X_T \subset X$ be the set of
leaves of the CPS $T$ and let $(x_1, x_2, \dots,  x_t)$ be some arbitrary but fixed ordering of them. Start with $x_1$, delete it and delete any reticulation node with
outdegree 0 and perform the standard cleaning-up operations\footnote{Deleting reticulation nodes with outdegree 0; suppressing nodes with indegree and outdegree both equal to
1; deleting leaves unlabelled by taxa; deleting nodes with indegree 0 and outdegree 1.} until the resulting network is a phylogenetic network. Repeat for $x_2$ and so on until
arriving at $x_t$ which is simply relabelled by the new taxon $x$. Let $N'$ be the resulting network. By construction, $N'$ displays all $T_i' \in \mathcal{T}'$ and $r(N') \leq
k$.
\end{proof}

Now to the common chain reduction rule:

\begin{claim}
Let $\mathcal{T}$ be a set of binary, unrooted trees on $X$ and let $\mathcal{T}'$ be the set of trees in $\mathcal{T}$ after a single application of the Common $5k$-Chain
reduction rule. Then, $h^{ru}(\mathcal{T}) \leq k \Leftrightarrow h^{ru}(\mathcal{T}') \leq k$.
\end{claim}

\begin{proof}
For the first direction (from the reduced to the original instance) let $C = (x_1, x_2, \dots, x_t  )$ be a subset of the taxa $X$ that defines a maximal common chain of length
$> 5k$.
Suppose that, in $\mathcal{T}$, we have clipped $C$ down to a reduced chain $C_R = (x_1,\ldots ,x_{5k})$. Let $\mathcal{T}_R$ be the set of these clipped (or reduced) trees and
let $N_R$ be a network that displays some rooted version of $\mathcal{T}_R$ with $k$ reticulation nodes.
Since the generator has at most $5k-1$ sides, there must exist at least one side containing at least two leaves of the chain. Let~$x_i$ and~$x_j$ be two leaves of the chain
that are on the same side~$s$ of the generator, with~$x_i$ above $x_j$. Clearly, this side must be an edge side. We will consider the case that $i<j$. The case that $j<i$ can
be handled symmetrically.

First suppose that~$\{i,j\}\neq\{1,2\}$. Then, we move all the taxa of the chain on the appropriate location on the side $s$ of $x_i,x_j$ of the generator $G$. We take all taxa
$x_{\ell}$ such that $\ell > j$ and plug them after $x_j$ in $s$, by appropriately subdividing the unique edge exiting the parent of $x_j$. We do the opposite for all the taxa
$x_{\ell'}$ such that $\ell' < i$ i.e., plug them ``before" $x_i$ in $s$ by appropriately subdividing the unique edge entering the parent of $x_i$.

Now suppose that~$i=1$ and~$j=2$.  Then we take any other pair of leaves that are on the same side of the generator and go back to the previous case. To see that such leaves
exist, assume that $\{x_1,x_2\}$ is the only pair of leaves that are on the same side. If the trees in $\mathcal{T}_R$ are not all identical, then there exists at least one
leaf~$y$ that is not in the reduced chain~$C_R$. Since the generator has at most~$5k-1$ sides, and the chain has~$5k$ leaves, this implies that each side contains at least one
leaf of the chain. Let~$x_q$ be a leaf of the chain that is on the same side as~$y$. This is only possible when $q\in\{1,5k\}$. If~$q=1$ this implies that the original
chain~$C$ was not maximal and we obtain a contradiction. If~$q=5k$, then we can add~$y$ to~$C_R$ and obtain a longer common chain~$C_R'$. Repeating this argument, we eventually
obtain a contradiction or find out that all trees in $\mathcal{T}_R$ are identical (a case that can be handled trivially).

As mentioned before,  the case that $j<i$ can be handled symmetrically. In this case, we make sure that $\{i,j\}\neq\{5k-1,5k\}$.

\medskip
\noindent \textbf{Expanding Step:} We still need to expand the chain by introducing the ``missing" taxa (the ones that disappeared after the clipping of the chain).  Move all
these taxa $\{ x_{5k+1}, \dots, x_{t} \}$ to the side $s$ in such a way that either~$C$ or the reverse sequence becomes a chain in the network.
In that way, from $N_R$ on $X_R$ (the leaf label set without the clipped labels after an application of the $5k$-CC rule) we have created a new network $N$ on $X$ with the same
reticulation number as $N_R$. We still need to argue that $N$ displays some appropriately rooted version of $\mathcal{T}$.

Take any tree $T_R \in \mathcal{T}_R$. Perform all the previous operations (applied on $N_R$) on $T_R$. In other words, move appropriately all the corresponding taxa on the
same side of the root and re-introduce the ``missing'' taxa in such a way that either~$C$ or the reverse sequence becomes a chain in the tree. In this way, from the rooted
network $N_R$ on $X_R$ that displays a rooted version of the truncated trees in~$\mathcal{T}_R$, we construct a new rooted network $N$ on $X$ \textit{and} rooted versions
$\mathcal{T}^{\rho}$ of the trees $\mathcal{T}$.

We now argue that~$N$ displays $\mathcal{T}^{\rho}$.  Let~$T^\rho\in \mathcal{T}^{\rho}$, let~$T_R$ be the corresponding reduced tree in~$\mathcal{T}_R$ and let~$T_R^\rho$ be a
rooting of the reduced tree~$T_R$ that is displayed by~$N_R$.

If neither~$x_1$ and~$x_2$ nor~$x_{5k-1}$ and~$x_{5k}$  have a common parent in~$T_R^\rho$, then it is clear from the construction that~$N$ displays~$T^\rho$. If~$x_1$
and~$x_2$ have a common parent in~$T_R^\rho$, then it is possible that~$x_1$ and~$x_2$ are on the same side of the generator of~$N_R$ with~$x_1$ above~$x_2$. If we moved all
chain-taxa below~$x_2$ on this side then this would inverse the chain, which would be a problem. However, this does not happen since in this case~$i<j$ and in that case we made
sure that~$\{i,j\}\neq\{1,2\}$. Symmetrically, if~$x_{5k-1}$ and~$x_{5k}$ have a common parent in~$T_R^\rho$ and~$x_{5k}$ is above~$x_{5k-1}$ on some side of~$N_R$, then we do
not move all taxa to this side because then~$j<i$ and hence~$\{i,j\}\neq\{5k,5k-1\}$.


\medskip
For the other direction (from the original instance to the truncated), let $N$ be a rooted network that displays some rooted version of $\mathcal{T}$ with $k$
reticulation nodes. Let $\mathcal{T}_R$ be the set of trees from $\mathcal{T}$ after a single application of the common chain reduction rule. Then, from $N$, we
will show how to create a rooted network $N'$ that displays some appropriately rooted version of $\mathcal{T}'$. Let $N'$ be the network obtained from
$N$ as follows: let $C$ be a common chain on $\mathcal{T}$ of length greater than $5k$. Take $C$ on $N$ and ``clip" it i.e., delete all leaves $x_\ell$ with
indexes $\ell \geq 5k+1$, and apply the usual cleaning-up steps. If the root of $N$ happens to be on
the chain then take the single edge $e$ entering the parent of the last surviving taxon of the chain with index $x_{5k}$, subdivide it and introduce the new root location at
the new intermediate node that subdivides $e$. Do the same on all $T \in \mathcal{T}$. Thus, we have created a new rooted network $N'$ and a rooting for all trees
in $\mathcal{T}_R$, all on $X'$ (without the ``excess" taxa deleted from the common chain). Obviously, by construction, the reticulation number of $N'$ has not
increased.  It remains to show that $N'$ displays $\mathcal{T}_R$ which follows immediately since $N$ displays (a rooted version of) $\mathcal{T}$.
\end{proof}

These two claims show that successive applications of the CPS and $5k$-CC rules do not change the hybridization number of the resulting reduced instances. Assuming that we have
applied these two rules as often as possible, let $\mathcal{T}'$ be the resulting instance. From the previous analysis we know that $\exists N'$ such that $r(N') \leq k$. Since
each common chain of $\mathcal{T}'$ has length $\leq 5k$, we conclude that in $N'$ we cannot find a chain of length greater than $5k$ where all leaves are on the same side of
the underlying generator (otherwise it would constitute a common chain and it would be clipped). Thus, $N'$ has at most $5k-1$ taxa on each \textit{edge} side and, obviously,
at most one taxon on each \textit{node} side. Thus, the total number of taxa that $N'$ can have is at most

$$
(5k-1) \cdot \underbrace{(4k-1)}_\text{\# of edge sides} +  \underbrace{k}_\text{\# of node sides} < 20k^2.
$$


The above kernelization eventually terminates: at each step we either identify a common pendant subtree or a long common chain or, if none of these is possible, we terminate.
Each reduction step reduces the number of taxa by at least 1, so we eventually terminate in polynomial time.\\
\\
The kernel we have described can be used to give an FPT algorithm to answer the question, ``Is $h^{ru}(\mathcal{T}) \leq k$?''. Let $\mathcal{T'}$ be the kernelized set of
trees. If the cardinality of the set of leaves given in the above bound is violated, we know that the answer is NO. So, assume it is not violated. We simply guess by
brute-force the root location of each tree in $\mathcal{T'}$. Each collection of guesses yields a set of rooted binary phylogenetic trees $\mathcal{T''}$, and we ask ``Is
$h^{r}(\mathcal{T''}) \leq k$?'' Clearly, the answer to ``Is $h^{ru}(\mathcal{T}) \leq k$?'' is YES if and only if at least one of the ``Is $h^{r}(\mathcal{T''}) \leq k$?''
queries answers YES. The kernelization procedure ensures that each tree in $\mathcal{T'}$ has $O(k^2)$ taxa and thus also $O(k^2)$ edges. Hence, the overall running time is the
time for the kernelization procedure plus $[O(k^2)]^t$ calls to an algorithm for {\hn}, where $t = |\mathcal{T}|$. Noting that $t \leq 2^k$ (otherwise the answer is trivially
NO), and that {\hn} is FPT \cite{vanIersel20161075}, we obtain an overall running time of $O( \text{poly}(n) +  f(k) \cdot \text{poly}(n) ) = O(f(k) \cdot \text{poly}(n))$.

This concludes the proof of Theorem \ref{thm:unrooted}.



\section{Conclusions and open problems}
\label{sec:conclusions}

In this article we have studied two variations of the classical hybridization number ({\hn}) problem: the root-uncertain variant {\ruhn} and the unrooted variant {\uhn}. We have also studied the natural unrooted variant of the tree containment ({\tc}) decision problem, {\utc}.

As we have seen, both {\tc} and {\utc} are NP-complete and FPT in reticulation number. The natural open question here is whether our FPT algorithm for {\utc}, with running time $O( 4^k \cdot \text{poly}(n) )$,
can be improved to achieve a running time of $O( 2^k \cdot \text{poly}(n) )$, which is trivial for {\tc}. Also, the {\tc} literature has not yet considered pre-processing, so it would be interesting to adapt our kernelization strategy to the rooted context.

Regarding {\hn} and {\ruhn}, both are APX-hard. It is known that if {\hn} (respectively, {\ruhn}) is in APX then so too is \textsc{DFVS}. However, at present we do not have a
reduction from {\ruhn} to {\hn}, which means that (from an approximation perspective) {\hn} might be easier than {\ruhn}. This is an interesting question for future research:
it remains a possibility that both {\hn} and \textsc{DFVS} are in APX, but {\ruhn} is not.  Both {\hn} and {\ruhn} are FPT in hybridization number, via a quadratic kernel. For {\ruhn} a
pertinent question is whether, in the case of just \emph{two} input trees, the best known FPT running time for {\hn} can be matched, which is $O(3.18^{k} \cdot \text{poly}(n))$
\cite{whidden2013fixed}. This raises the question of whether, and in how far, the successful agreement forest abstraction can be adapted for {\ruhn}.

In terms of approximation the other variant of {\hn}, {\uhn}, differs quite strikingly from {\hn}, although we note that in this article we have only studied {\uhn} on two
trees. For two trees {\ruhn} is (due to its equivalence with  \textsc{TBR}) in APX, while it is still unknown whether {\hn} is in APX (see the above discussion). This gap in
approximability is similar to that which exists between \textsc{Maximum Acylic Agremeent Forest (MAAF)} and \textsc{Maximum Agreement Forest (MAF)} on two rooted trees
\cite{approximationHN}. This is not so surprising given that \textsc{MAAF} is essentially equivalent to {\hn}, and both the rooted and unrooted variants of \textsc{MAF} (which
are essentially equvalent to r\textsc{SPR} and \textsc{TBR} respectively) are firmly in APX.

Alongside the complexity discussions above it's tempting to ask which of the problems studied in this article can (in some formal sense) be ``reduced'' to each other. The
APX-hardness reduction already shows that {\hn} can be reduced to {\ruhn} in a highly approximation-preserving way. Can {\ruhn} be reduced to {\hn}? Can {\hn} be reduced to
{\uhn}? Can {\ruhn} be reduced to {\uhn}?

\section{Acknowledgments}
Leo van Iersel was partly funded by a Vidi grant from the Netherlands Organisation for Scientific Research (NWO) and by the 4TU Applied Mathematics Institute. The work of both
Olivier Boes and Georgios Stamoulis was supported by a NWO TOP 2 grant.

\bibliographystyle{plain}
\bibliography{bibliographyTOP}

\clearpage
\appendix

\section{Omitted APX-hardness proof}

Here we give the details of the omitted proof that {\ruhn} is APX-hard (Theorem \ref{thm:apxhard}).

\begin{proof}
To establish the result we provide
a linear (L) reduction \cite{lreduc}  from {\hn} to {\ruhn}. Formally, an L-reduction is defined as follows:

\begin{dfn}
Let $A,B$ be two optimization problems and $c_A$ and $c_B$ their respective cost functions. A pair of functions $f,g$, both computable in polynomial time, constitute an
L-reduction from $A$ to $B$  if the following conditions are true:

\begin{enumerate}
\item $x$ is any instance of $A$ $\Rightarrow$ $f(x)$ is an instance to $B$,
\item $y$ is a feasible solution of $f(x)$ $\Rightarrow$ $g(y)$ is a feasible solution of $x$,
\item $\exists \alpha \geq 0$ such that $OPT_B(f(x)) \leq \alpha OPT_A(x)$,
\item $\exists \beta \geq 0$ such that for every feasible solution $y'$ for $f(x)$ we have $|OPT_A(x) - c_A(g(y'))| \leq \beta |OPT_B(f(x)) - c_B(y')|$
\end{enumerate}
where $OPT_{A}$ is the optimal solution value of problem $A$ and similarly for $B$.
\end{dfn}

{\hn} is already known to be NP-hard and APX-hard for $|\mathcal{T}|=2$. We will give a $(\alpha, \beta)=(2,1)$ L-reduction from {\hn} to {\ruhn}.

Let $(T_1, T_2)$ be the two trees that are given as input to {\hn}, and let $X$ be their set of taxa, where $n=|X|$. Let $(T^{*}_1, T^{*}_2)$ be the unrooted trees constructed
in the proof of Lemma \ref{lem:similar}, except that here we make the ``$c$'' and ``$d$'' caterpillars slightly longer: length $(2n+3)$ instead of length $(n+1)$. It is easy to
check that with these longer caterpillars (\ref{eq:plus1}) still holds.

To avoid technical complications (the approximation oracle outputting an exponentially large answer) it is helpful in this proof to assume that, when given two unrooted binary
trees as input, each on $n'$ taxa, an approximation oracle for {\ruhn} never returns a network $N$ such that $r(N) > n'$. This is reasonable because one can always root the two
trees in arbitrary places and construct a trivial network with $n'$ reticulation nodes that displays both rootings, simply by ``merging'' the two rooted trees at their taxa and
adding a new root.

The reduction proceeds as follows. First we check whether $h^{r}( T_1,T_2 ) = 0$. This can easily be performed in polynomial time since this holds if and only if $T_1$ and
$T_2$ are isomorphic. If the equality holds, then there is no need to call the {\ruhn} approximation oracle: simply return $T_1$ as an optimal solution to {\hn} (i.e. $T_1$ is
a rooted phylogenetic network that trivially displays both $T_1$ and $T_2$). Otherwise, we know $h^{r}(T_1, T_2) \geq 1$. Given that (as shown in the previous proof)
$h^{ru}(T^{*}_1,T^{*}_2 ) = h^{r}(T_1, T_2) + 1$ it then follows immediately that $h^{ru}(T^{*}_1,T^{*}_2 ) \leq \alpha \cdot h^{r}(T_1, T_2)$, thus satisfying the ``forward
mapping'' part of the L-reduction (i.e. point 3 of the definition).

We now consider the ``back mapping'' part, i.e. point 4 of the definition. Consider the solution returned by the approximation oracle for {\ruhn}. We distinguish two cases.
First, suppose the solution given by the oracle roots $T^{*}_1$ and $T^{*}_2$ in a ``stupid'' way i.e. such that two oppositely oriented rooted caterpillars are induced. If
this happens, the network given by the oracle will have reticulation number at least $2n+1$ (i.e. the length of the caterpillar minus two). We know $h^{ru}(T^{*}_1,T^{*}_2 ) =
h^{r}( T_1,T_2 ) + 1$  so, trivially, $h^{ru}(T^{*}_1,T^{*}_2 ) \leq n+1$. Hence, the additive gap between the quality of the approximate and exact solution to
$\ruhn(T^{*}_1,T^{*}_2)$ is at least $(2n+1)-(n+1) = n$. To complete the L-reduction in this case, we simply discard the output of the approximation oracle, and return a
trivial solution to $\hn$ with $n$ reticulations (e.g. by merging $T_1$ and $T_2$ at their taxa and adding a new root). Given that $h^{r}(T_1,T_2 ) \geq 1$, the additive gap
between our approximate solution to $\hn$ and its true optimum is at most $n-1$. Clearly, $n-1 \leq \beta \cdot n$, as required by the ``back mapping'' part of the L-reduction.

The other case is if the the approximate solution roots the trees in a ``sensible'' way. That is, without loss of generality, $T^{*}_1$ is rooted somewhere inside its $c$-part and $T^{*}_2$ is
rooted somewhere above its $T_2$ part (i.e. somewhere in its $c$-part or $d$-part). This case is somewhat more subtle. Suppose the network $N'$ returned by the approximation oracle for {\ruhn} has reticulation number $q$.

We restrict $N'$ to only those nodes and edges necessary to display $T_1$ and $T_2$. The images of $T_1$ and $T_2$ inside $N'$ can easily be found in polynomial time because
(a) images of the rootings of $T^{*}_1$ and $T^{*}_2$ within $N'$ will be returned by the oracle as certificates, and (b) due to the location of the chosen roots, images of
$T_1$ and $T_2$ within $N'$ can be directly extracted from the images of the rootings of $T^{*}_1$ and $T^{*}_2$.
(Note that if $N^{T_1, T_2}$, which simply denotes the union of the two images, does not comply with the degree restrictions of a rooted binary phylogenetic network, it is easy
to modify it so that it does comply with these restrictions, without raising its reticulation number). $N^{T_1, T_2}$ displays $T_1$ and $T_2$, and has $X$ as its set of taxa:
this will be the solution we return for {\hn}.

Suppose that $N^{T_1, T_2}$ has reticulation number $p$. Given that the restriction operation does not create new reticulation nodes (but possibly deletes them), $p \leq q$. To make the reduction go through, we need:
 \[
|p - h^{r}(T_1,T_2)| \leq \beta | q - h^r(T^{*}_1, T^{*}_2) |
 \]
which is equivalent to,
\[
|p - h^{r}(T_1, T_2)| \leq | q - h^{r}(T_1, T_2) - 1 |
\]

 If $p \leq (q-1)$ this is immediate. However, what if $p=q$? This occurs if $N^{T_1, T_2}$ has just as many reticulations as $N'$ i.e. the restriction process did not delete any reticulation nodes. However, this cannot
happen, for the following reason. Consider again the image of the rooting of $T^{*}_1$ in $N'$. Extract all nodes and edges on the image that minimally span the taxa $\{c_0,
..., c_{2n+2}, d_0, ..., d_{2n+2}\}$.  Do this also for the image of the rooting of $T^{*}_2$ in $N'$. Let $N^{Cat}$ be the network obtained from $N'$ by restricting to the
union of these nodes and edges.  $N^{Cat}$ is a rooted binary phylogenetic network on $\{c_0, ..., c_{2n+2}, d_0, ..., d_{2n+2}\}$ that displays the caterpillars induced by the
rootings of $T^{*}_1$ and $T^{*}_2$. Crucially, $r(N^{Cat}) > 0$. This is because, as observed earlier, the caterpillar parts of $T^{*}_1$ and $T^{*}_2$ remain non-isomorphic
wherever you place the root exactly. In turn this means that $N^{Cat}$ must contain at least one reticulation node $v$. This reticulation node $v$ necessarily lies on the
images of both $T^{*}_1$ and $T^{*}_2$, since otherwise it would not have survived to be a reticulation node in $N^{Cat}$ (i.e. at most one of its two incoming edges would have
been extracted, and hence it would have been suppressed). However, $v$ does not lie on the part of the image of $T^{*}_1$ that contributed to $N^{T_1, T_2}$, since $v$ is not
part of the image of $T_1$. Hence, $v$ lies on (at most) the image of $T_2$, meaning that $v$ is a reticulation node that does not survive when $N^{T_1, T_2}$ is created
(because at most one of its two incoming edges is extracted), and hence $p < q$.
\end{proof}

The above proof can, if desired, be strengthened to an $(\alpha, 1)$ L-reduction for any constant $1 < \alpha < 2$. To achieve this, we do not begin by checking whether $h(T_1,
T_2)=0$, but instead whether $h(T_1, T_2) < \frac{1}{\alpha-1}$, and if so we simply return an optimal solution. For constant $\alpha$ this can be performed in polynomial time,
because computation of {\hn} is FPT in $h^{r}(T_1, T_2)$.

\end{document}